\documentclass[alpha-refs]{wiley-article}

\usepackage{color}
\newcommand{\indep}{\rotatebox[origin=c]{90}{$\models$}}

\usepackage{ulem, subcaption, physics, xcolor, enumitem}

\usepackage[safe]{tipa} %for omega 

\newcommand{\textp}{\textit{p}\ }

\newcommand{\calN}{\mathcal{N}}

\newcommand{\calP}{\mathcal{P}}
\usepackage{siunitx}

\papertype{Original Article}
\paperfield{Journal Section}

\title{A Bayesian decision support system for counteracting activities of terrorist groups}

\author[1\authfn{1}]{Aditi Shenvi}
\author[3\authfn{1}]{F. Oliver Bunnin}
\author[1,2]{Jim Q. Smith}

\contrib[\authfn{1}]{Equally contributing authors.}

\affil[1]{Department of Statistics, University of Warwick, Coventry, UK}
\affil[2]{The Alan Turing Institute, London, UK}
\affil[3]{Natwest Markets, London, UK}

\corraddress{A. Shenvi, Department of Statistics, University of Warwick, Coventry, CV4 7AL, UK}
\corremail{aditi.shenvi@warwick.ac.uk}

\fundinginfo{AS was supported by the University of Warwick Chancellor's International PhD Scholarship and the Alan Turing Institute's PhD Enrichment Award whilst working on this project during her PhD. JQS was supported by the Alan Turing Institute and funded by the EPSRC [grant number EP/K03 9628/1]. FOB was funded by The Alan Turing Institute Defence and Security Project G027 during 2019 and 2020 whilst working there.}

\runningauthor{Shenvi et al.}

\begin{document}

\maketitle

\begin{abstract}
Activities of terrorist groups present a serious threat to the security and well-being of the general public. Counterterrorism authorities aim to identify and frustrate the plans of terrorist groups before they are put into action. Whilst the activities of terrorist groups are likely to be hidden and disguised, the members of such groups need to communicate and coordinate to organise their activities. Such observable behaviour and communications data can be utilised by the authorities to estimate the threat posed by a terrorist group. However, to be credible, any such statistical model needs to fold in the level of threat posed by each member of the group. Unlike in other benign forms of social networks, considering the members of terrorist groups as exchangeable gives an incomplete picture of the combined capacity of the group to do harm. Here we develop a Bayesian integrating decision support system that can bring together information relating to each of the members of a terrorist group as well as the combined activities of the group. 
% Please include a maximum of seven keywords
\end{abstract}

%%%%%%%%%%%%%%%%%%%%%%%%%%%%%%%%%%%%%%%%%%%%%%%%%%%%%%%%%%%%%%%%%%
%%%%%%%%%%%%%%%%%%%%%%%%%% INTRODUCTION %%%%%%%%%%%%%%%%%%%%%%%%%%
%%%%%%%%%%%%%%%%%%%%%%%%%%%%%%%%%%%%%%%%%%%%%%%%%%%%%%%%%%%%%%%%%%

\section{Introduction} \label{sec:introduction}

% Here include a review of the relevant research on the application. What is the point of doing this work for this specific application? Where do the existing methods fall short?
% Include a proper discussion of the need for such work in the UK and abroad. Perhaps mention recent interview by MI5 head and/or the Afghanistan Karzai Airport terrorist attack. 
% thousands of potential leads coming in every month (Ken McCallum's interview)... "difficult decisions based on fragmentary information"
% Describe need for an integrating decision support system here
% Do not include review of SNA methods for terrorism as this is the introduction for the entire two-part mode and not specifically the network aspect
% MAKE DISTINCTION BETWEEN "TERRORIST ACTIVITIES" AND "PREOPERATIONAL ACTIVITIES".

Terrorism persists all over the world. Attacks perpetrated by terrorist groups, both large and small, have the potential to endanger lives and cause serious damage to the security and well-being of the general public. \citet{OHCHR32} states ``Terrorism clearly has a very real and direct impact on human rights, with devastating consequences for the enjoyment of the right to life, liberty and physical integrity of victims. In addition to these individual costs, terrorism can destabilize Governments, undermine civil society, jeopardize peace and security, and threaten social and economic development". Hence, developing effective counter-terrorism strategies are of high relevance for any Government. The UK's Strategy for Countering Terrorism, known as CONTEST, outlines its counter-terrorism framework built on the four `P' work strands described in \citet{contest} as follows: 
\begin{enumerate}[itemsep=0em, label=(\alph*)]
\item Prevent: to stop people becoming terrorists or supporting terrorism.
\item Pursue: to stop terrorist attacks.
\item Protect: to strengthen protection against a terrorist attack.
\item Prepare: to mitigate the impact of a terrorist attack.
\end{enumerate}

\noindent The objectives of the ``Pursue" work strand are to detect, understand, investigate, and disrupt terrorist attacks \citep{contest}. Counter-terrorism authorities can closely monitor the (pre-incident) activities\footnote{Throughout this paper, we use ``activities" to refer to preparatory or pre-incident activities of terrorists, and ``attack" to refer to a terrorist plot before or after execution.} of suspected terrorist groups -- once appropriate legal permissions have been secured\footnote{The Regulation of Investigatory Powers Act 2000 governs the methods of information gathering. A warrant to gather information requires permission from the Secretary of State and an independent senior judge. They approve the warrant only when they are convinced that it is necessary and proportionate \citep{ripa, chorley_2021}} -- in an attempt to frustrate any potential attacks for which they might have been preparing. Whilst groups intent on terrorism may very likely attempt to hide or disguise their intentions and activities to avoid detection and scrutiny, they still need to perform certain preparatory tasks and to communicate with each other in order to plan, organise and execute a joint terrorist attack. For instance, in the aftermath of the fall of Kabul to the Tabilan and prior to the suicide bombing attack at Kabul's Hamid Karzai International Airport, there was credible intelligence of increased threat from a terrorist attack at the airport during evacuation efforts by foreign governments \citep{Afgh_aljazeera, Afgh_reuters}. Similarly, there were several warning signs of a possible attack prior to the September 11, 2001 terrorist attacks perpetrated by al-Qaeda in the United States. These signs include intercepted phone calls in Yemen detailing plans for an al-Qaeda summit in Malaysia in December 2009, Zacarias Moussaoui (a member of al-Qaeda) enrolling in flight training in Minnesota, and arrival of key al-Qaeda members in the United States from January 2000 through April 2001 \citep{911}. 

Governments worldwide spend a considerable portion of their funds on counter-terrorism efforts. In the UK, the government committed to a 30\% increase in counter-terrorism spending from \pounds 11.7 billion to \pounds 15.1 billion from 2015 to 2020 \citep{funding_CT_UK}. In June 2021, the UK unveiled the first-phase of a new state-of-the-art Counter Terrorism Operations Centre (CTOC) \citep{CTOC}. The United States of America, on the other hand, spent \$2.8 trillion during fiscal years 2002 through 2017 on homeland security, international security programs, and the wars on terrorism in Afghanistan, Iraq, and Syria \citep{funding_CT_USA}. Despite the staggering amounts of money invested into counter-terrorism efforts, the human, time, and material resources available for these purposes are, of course, not unlimited. Every week hundreds of new leads arise in the form of partial and fragmentary information \citep{Anderson2016}. MI5 director Ken McCallum stated in a recent interview \citep{chorley_2021} that ``difficult decisions [have to be made] based on fragmentary information". Hence, counter-terrorism authorities must decide how to use the limited resources available to them to prioritise and de-prioritise cases according to the incoming noisy leads. Additionally, there are some individuals about whom authorities are well-informed while there are others whose behaviours and intentions remain largely unknown. 

Statistical models of observable terrorist behaviours and activities (for example, \citet{raghavan2013hidden, gruenewald2019suspicious, allanach2004detecting, singh2004stochastic}) typically fail to take into account the different roles played by the members within the group or the interactions between them. On the other hand, statistical techniques from social network analysis (SNA) when applied to terrorist networks (see review in Section \ref{subsec:related_work}), take into account the interactions of the members of the group but typically also treat the individuals as exchangeable. Each individual member of the group has their own role to play in the activities and overall mission of the group \citep{roles}. However, although it might seem ideal in this context for a model to explicitly accommodate the various roles within a terrorist group, these depend on the type of terrorist attack being planned and more importantly, on the complex dynamics of the group that are not trivial to uncover \citep{roles}. We note here that although there exist instances of role identification within terrorist groups (for example,  \citet{shaikh2007graph, qin2005analyzing}), these often rely on measures such as centrality within the terrorist network to identify roles rather than the activity data concerning these individuals. An exception is \citet{aitkin2017statistical} which uses activity data but, in its current form, cannot model the evolution of the terrorist network.

\citet{bunnin2019bayesian} presented a Bayesian hierarchical model for the progression of a lone terrorist suspect's threat state (e.g. the threat states could be defined as `Neutral/Non-threatening', `Mobilised', `Preparing', `Training' and `Active/Threatening')  based on streaming data, hereafter called the \textit{Lone Terrorist} model. This model uses incoming patchy and noisy data to infer which preparatory tasks the suspect might be engaging in, and thereby provides a probabilistic estimate of which threat state they currently occupy. To study terrorist groups, however, it is insufficient to perform a multivariate extension of the Lone Terrorist model as doing so would treat the group as a collective unit, only considering the preparatory tasks and the threat state that the group collectively engages in and occupies. Such a model would not take into account the interactions between members or the different roles played by them. To incorporate these essential aspects, we demonstrate how the Lone Terrorist model can be combined with a statistical network model.   

In this paper, we present a Bayesian integrating decision support system (IDSS) designed to aid in monitoring the threat presented by known or suspected terrorist groups. An IDSS is a computer-based statistical tool that supports complex decision-making by enabling decision-makers to explore the implications of the various options available to them \citep{leonelli2015bayesian, barons2020decision} (e.g. the effect on a terrorist group's preparedness by convincing one of its member to leave or act against the group \citep{chorley_2021}). It does so by combining together different models that integrate expert knowledge and judgement to enable reasoning about distinct aspects of a complex system, such as the threat posed by a terrorist group. The novel form of IDSS we propose here combines together two main components (1) a Lone Terrorist model for each member of a known or suspected terrorist group, and (2) a dynamic weighted \textit{Terrorist Network} that models the pairwise interactions between the members of the group. The Lone Terrorist model acts as a proxy for the role of each member within the group by modelling the threat state occupied by that member assuming that they were carrying out a terrorist attack by themselves. Such an IDSS can then be used create informative indicators of the imminence of an attack by a terrorist group. 

Section \ref{sec:observable_data} describes the observable data on the behaviours and activities of the members of a group that counter-terrorism authorities might have access to after acquiring the necessary legal permissions. In Section \ref{sec:background}, the Lone Terrorist model and the IDSS framework are briefly reviewed. Section \ref{sec:dynamic_network} presents the Terrorist Network which models the pairwise interactions between members of a terrorist group and also includes a review of relevant SNA methodologies applied to terrorist networks. Section \ref{sec:idss_for_terrorism} combines together the two components of the proposed IDSS and describes how this IDSS could be used to monitor measures which act as indicators of terrorist attacks. Section \ref{sec:analysis} presents a proof of concept application of the proposed IDSS. We conclude the paper with a discussion of our contributions and the avenues for future research leading from our work in Section \ref{sec:discussion}. 

%%%%%%%%%%%%%%%%%%%%%%%%%%%%%%%%%%%%%%%%%%%%%%%%%%%%%%%%%%%%%%%%%%
%%%%%%%%%%%%%%%%%%%%%%%%% OBSERVABLE DATA %%%%%%%%%%%%%%%%%%%%%%%%
%%%%%%%%%%%%%%%%%%%%%%%%%%%%%%%%%%%%%%%%%%%%%%%%%%%%%%%%%%%%%%%%%%

\section{Observable Data} \label{sec:observable_data}

% Add information about what individual level data might be observed by the policing authorities

% Add here details of the pairwise communications data that could be observed for making the network and also for estimating the edge weights

The Lone Terrorist and Terrorist Network models are designed to be informed by routinely and covertly observable data on suspected terrorists that counter-terrorism authorities are able to access after acquiring the required legal permissions, and only when necessary and proportionate (see Section \ref{sec:introduction}). According to \citet{ripa2}, ``Surveillance, for the purposes of RIPA [The Regulation of Investigatory Powers Act 2000 \citep{ripa}], includes monitoring, observing or listening to persons, their movements, conversations or other activities and communications." Thus, counter-terrorism authorities may receive information from various sources \citep{chorley_2021}. Example sources include monitoring of physical meetings, interception of electronic communications, and intelligence obtained from other policing agencies, covert informants, or the public.

In particular, terrorists working together within a group need to communicate to coordinate their joint efforts. These communications give rise to observable data. There are at least five types of potentially knowable or observable data which indicate ties between suspected terrorists that can be obtained by the counter-terrorism authorities:
\begin{enumerate}[itemsep=0em,  label=(\alph*)]
\item Existing kinship or social links;
\item Work or other shared affiliations;
\item Bilateral electronic communications (e.g. telephone, email, Whatsapp etc);
\item Physical meetings (observed directly or through closed circuit television);
\item Financial transactions (e.g. bank transfers between accounts).
\end{enumerate}
The first two items are relatively static whereas the other three are more dynamic. Moreover, the first two do not necessarily indicate ties that are malicious in nature, but may enable a pre-existing tie that facilitates collaboration once other factors have come into play. Examples of such ties are the school and social ties that existed between several of the al-Qaeda September 11, 2001 terrorists \citep{Krebs2002}, the kinship tie between Saleem and Hashem Abedi -- the former being the suicide bomber of the May 22, 2017 Manchester Arena bombing and the latter his brother who was found guilty of aiding Saleem \citep{guardian_17mar2020} -- and the community ties surveyed by the US army in Thai villages in 1965 \citep{Meter2002}. 

Note that it is important to differentiate two types of data associated with communications: the content of such communications and ``secondary data", i.e. metadata such as the identities of parties and the timing, location and duration of communications. Often secondary data is available whilst content data is unavailable due to either encryption or limits prescribed by certain interception warrants. Moreover due to technology companies' planned future adoption of encryption for a wider range of communication technologies, the availability to investigators of content data is likely to decrease \citep{watney2020law, chorley_2021}. However, even secondary data without content data has proven to be extremely useful: ``so-called \textit{secondary data} can enable the tracing of contacts, associations, habits and preferences" \citep{Anderson2016}. We assume at a minimum some availability of secondary data. 

%%%%%%%%%%%%%%%%%%%%%%%%%%%%%%%%%%%%%%%%%%%%%%%%%%%%%%%%%%%%%%%%%%
%%%%%%%%%%%%%%%%%%%%%%%%%%% BACKGROUND %%%%%%%%%%%%%%%%%%%%%%%%%%%
%%%%%%%%%%%%%%%%%%%%%%%%%%%%%%%%%%%%%%%%%%%%%%%%%%%%%%%%%%%%%%%%%%

\section{Background: Two Components, One Decision Support System} \label{sec:background}

%In Section \ref{sec:introduction} we proposed an IDSS that combines together two distinct components; the first of which is a collection of Lone Terrorist models -- one for each member of a known or suspected terrorist group, and the second is the Terrorist Network model which will be described in Section \ref{sec:dynamic_network}. 

%In this section, we review the technical underpinnings of the Lone Terrorist model introduced in \citet{bunnin2019bayesian} and the relevant theory of the IDSS framework.

\subsection{Modelling Activities of Lone Terrorists} \label{subsec:lone_terrorist}
% Include here the review of the RVE model but call it something different and focus on the application part of the work
% Add a top to bottom description of the RVE
% Not sure whether the recurrences need to be added here at all

The Lone Terrorist model \citep{bunnin2019bayesian} is a three-level hierarchical Bayesian model developed to support the counter-terrorism authorities in their pursuit of terrorists and violent criminals acting alone to commit a violent crime against the general public. Since the counter-terrorism authorities get more leads every week than they can pursue, we define an open population of all suspects or persons of interest (POIs) at time $t$ as $\calP_t^*$ and the subset of $\calP_t^*$ that the authorities have decided to investigate and monitor at time $t$ as $\calP_t$. A bottom-up description of the three levels of the Lone Terrorist model for a suspect \textp $\in \calP_t$ is given below.

\subparagraph{Bottom Level} This level consists of a discrete time graphical model (more precisely, a reduced dynamic chain event graph, see references in \citet{bunnin2019bayesian}). This graphical model can be customised to a specific type of terrorist attack (e.g. a knife-attack). The states (also called here as ``threat states") of the model, which form the vertices of its associated graph, represent the possible paths of progression for the modelled terrorist attack. \citet{smithassault2018} provides several types of categorisations for a wide range of criminal behaviours which can be used to inform the threat states of the graphical model. Alternatively, these states can be more generically defined (e.g. `Mobilised', `Preparing', `Training' and `Active/Threatening'). In both cases, the model also includes a ``Neutral/Non-threatening" state which is an absorbing state representing that the suspect no longer presents a threat to the general public within the jurisdiction of the counter-terrorism authority. Denote by $X_t$ the latent random variable indicating the threat state occupied by a suspect \textp at time $t \geq 0$. The sample space of $X_t$ is given by the vertices $\{x_0, x_1, \ldots, x_n\}$ of the graph. Let $\pmb{\pi}_{t} = \{\pi_{t0}, \pi_{t1}, \ldots, \pi_{tn}\}$ where $\pi_{ti}$ indicates the probability of the suspect being in threat state $x_i$ at time $t$ for $i \in \{0, 1, \ldots, n\}$. 

\subparagraph{Intermediate Level} At this level, we define a collection of $R$ tasks associated with the threat states of the graphical model. At any time $t \geq 0$, denote the task vector by $\pmb{\vartheta}_t = \{\vartheta_{t1}, \vartheta_{t2}, \ldots, \vartheta_{tR}\}$ where each $\vartheta_{tj}$ is an indicator variable such that $\vartheta_{tj} = 1$ if \textp is enacting task $j$ at time $t$, for $j \in \{ 1, 2, \ldots, R\}$. Each task can be associated with one or more threat states of the graphical model. The purpose of the task vector is to enable the counter-terrorism authorities to estimate how far along the suspect is in their progression towards a specified or unspecified terrorist attack. 

\subparagraph{Surface Level} This level consists of the data $\{\textbf{Y}_t\}_{t \geq 0}$ relating to activities of the suspect \textp. This data may range from complete and reliable intelligence to partial and patchy secondary data. However, even noisy signals obtained from partial data can help the police to condition on the limited information under the Bayesian Lone Terrorist model and revise their judgements accordingly. For each task $\vartheta_{tj}$ in the intermediate level, we can associate a subset $Y_{tj} \subseteq \textbf{Y}_t$ of the data stream observed which informs whether \textp is engaged in task $\vartheta_{tj}$, for $j \in \{1,2, \ldots, R\}$ at time $t$. If the data is noisy, a \textit{filter function}\footnote{A filter function is simply a suitable function $\tau_j(\cdot)$ of the data $Y_{tj}$; see \citet{bunnin2019bayesian} for further details on suitable filter functions.} may be used to obtain some viable signal $Z_{tj}$ from the noisy data subset $Y_{tj}$. Denote by $\{\textbf{Z}_t\} = (Z_{t1}, Z_{t2}, \ldots, Z_{tR})$. 

%supplementary material cross-reference in paragraph below
The recurrences associated with the progressions in the Lone Terrorist model are described in Appendix A. At each time $t \geq 0$ the model provides as output the posterior probability vector $\pmb{\pi}_{t}$ associated with the suspect occupying one of threat states in the underlying model.  

\begin{example}(A hypothesised Lone Terrorist model)
A general template for a lone terrorist attack can be represented by the graph in Figure \ref{fig:rdceg}. The threat states of this model are represented by the vertices of this graph and the edges represent the possible transitions between the threat states. An example of the threat states, tasks and observable data for a gun attack are shown in Table \ref{tab:statespace}. 

\begin{figure}[h]
    \centering
    \includegraphics[scale = 0.4, trim = {0cm, 1cm, 0cm, 0.75cm}]{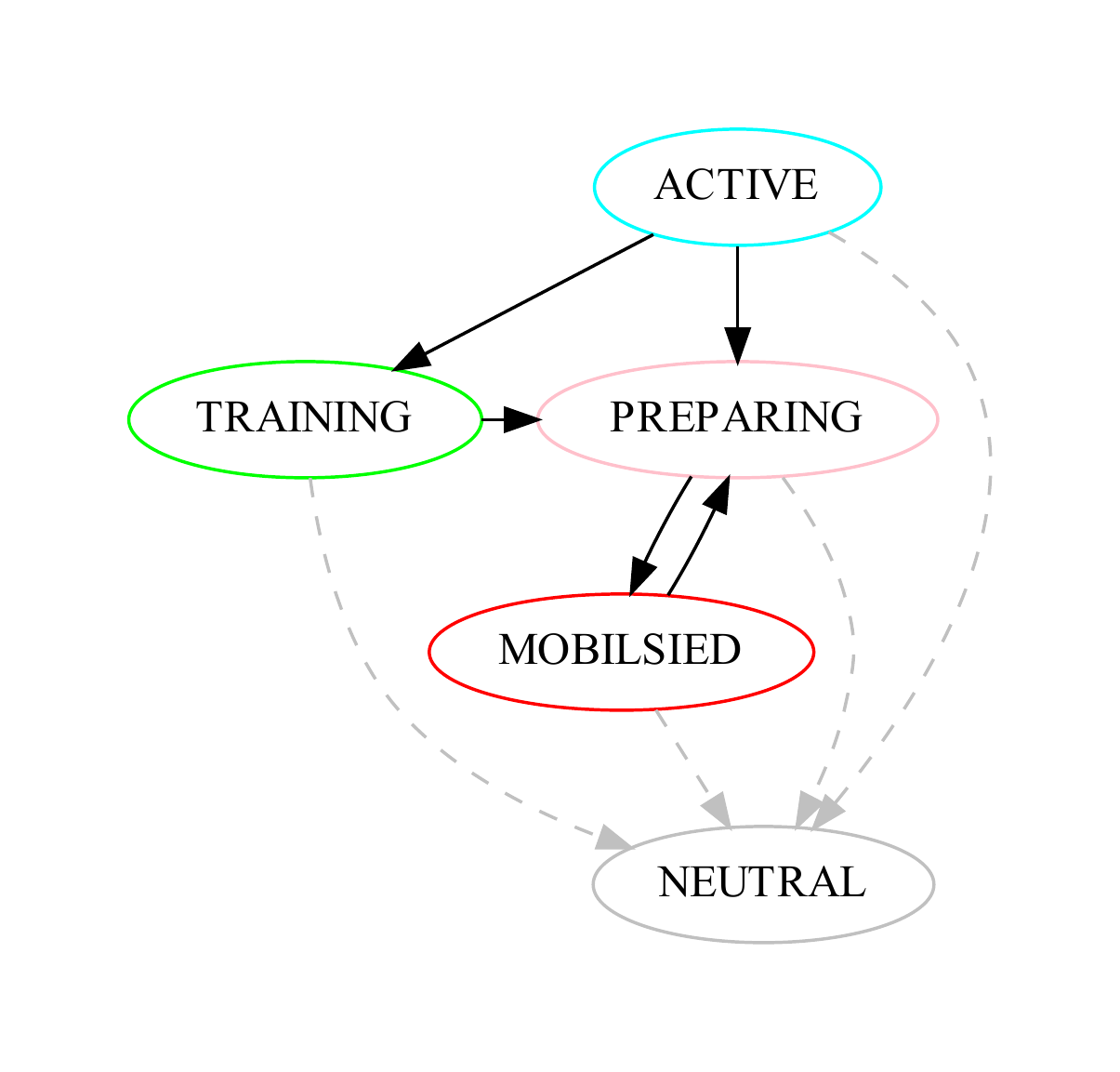}
    \caption{A graph for the Lone Terrorist model showing a generic set of threat states.}
    \label{fig:rdceg}
\end{figure}
\end{example}

\vspace{-0.65cm}
\begin{table}[h]
\caption{Example state space, task vector and observable data.}
\label{tab:statespace}
\begin{threeparttable}
\begin{tabular}{lll}
\headrow
\thead{Threat States} & \thead{Tasks} & \thead{Observable Data}\\
Active &    Engage with radicalisers &     Physically met with radicals \\
Training &  Make personal threats &   Personal threats made   \\
Preparing  & Learn to drive  &   Obtained driving license    \\
Mobilised &  Obtain vehicle  & Rented car \\
Neutral  & Engage in public threats   &  Public threats made on social media  \\
&   Obtain Financial Resources &   Sold assets \\
& Learn how to use a gun & Visited shooting ranges\\
& Acquire a gun &  Been seen with a gun\\
& Acquire ammunition & Met with gun \& ammunition dealer \\
& Reconnoitre targets &  Visits to target location made \\ 
& & Increase in finances\\
& & Reduced contact with family\\
& & Meetings with trained radicals\\
\hline
\end{tabular}
\end{threeparttable}
\end{table}

\vspace{-0.24cm}

\subsection{Integrating Decision Support Systems}
\label{subsec:IDSS}

% Add here the review of IDSS 
% Do we need a review of applications of the IDSS? 

An IDSS is a Bayesian unifying and integrating framework that combines component decision support systems -- each supporting decision-making about a distinct aspect of a complex system -- into a single entity (see \citet{leonelli2015bayesian, smith2015coherent, barons2020decision} and references therein). The transparent and statistically grounded framework of the IDSS enables a statistician to formally incorporate the judgements and uncertainties of the domain experts and decision-makers into the IDSS. Any available data is then fed through the relevant components of the IDSS with full consideration of these judgements and uncertainties. The outputs of all the components are then combined -- in a manner that is appropriate for the application -- to enable the decision-makers to fully evaluate the effects of any potential policies on their outcomes of interest. Thus, the IDSS systematically utilises all the relevant evidence to facilitate informed decision-making. The IDSS framework is particularly useful for decision-making for complex and evolving systems where the decision-makers need to simultaneously consider the effect of potential policies on several different evolving variables which can influence the outcomes of interest. In the case of terrorism, the decision-makers are regularly faced with making very difficult choices between which cases to prioritise and which to de-prioritise \citep{chorley_2021, Anderson2016}. Further, when dealing with terrorist groups, they need to not only consider the dynamics of the group but also the motivations, skills and preparedness of its individual members. 

In most cases, each component of the IDSS is itself a complex system that can be modelled by any appropriate Bayesian statistical model. Moreover, the outputs from one component may feed into other components as inputs. Under the IDSS framework, the outputs of the distinct components can be formally combined using another suitable \textit{composite model}. This suitable composite model must satify certain sufficient conditions to ensure that the inference made from the IDSS is coherent\footnote{\citet{leonelli2015bayesian, smith2015coherent} describe the sufficient conditions leading to a coherent integrating system.}. \citet{leonelli2015bayesian} provide several examples of suitable composite models such as Bayesian networks, chain event graphs, Markov networks, multiregression dynamic models (MDMs) and influence diagrams. For our application, we repurpose the decoupling methodology introduced for MDMs \citep{QueenandSmith93} to formally combine together, at each time $t \geq 0$, the outputs of the Terrorist Network and the individual Lone Terrorist models for each individual \textp in the population $\calP_t$ whilst estimating the parameters of each model independently. Below we briefly review the technical details of this methodology.

\vspace{-0.15cm}

\subsubsection{The Decoupling Methodology}
\label{subsubsec:Decoupling}

MDMs are a family of graphical models customised for modelling multivariate time-series data. The decoupling methodology of the MDM involves linking the components of a multivariate time-series data through their observations directly rather than through their parameter vectors. This allows the MDM to decompose the multivariate problem into a collection of univariate dynamic linear models \citep{WestandHarrison}. MDMs have had success in modelling multivariate time-series data in several diverse applications, see for example \citet{costa2015searching, anacleto2013multivariate, WilkersonSmith, barons2020decision}. Below we review the key results from the MDM theory. 
%freeman2011dynamic

Denote by $\textbf{Y}_t = \{Y_t(1), Y_t(2), \ldots, Y_t(n)\}$ a multivariate time-series composed of $n$ components at time ${t > 0}$. Let $\textbf{y}_t$ be the vector of observed values of all components of $\textbf{Y}_t$ and $y_t(i)$ be the observed value of component $Y_{t}(i)$. Further, let $\textbf{y}^t = \{\textbf{y}_1, \textbf{y}_2, \ldots, \textbf{y}_t\}$, $\textbf{y}^t_i = \{y_1(i), y_2(i), \ldots, y_t(i)\}$ and $\textbf{y}^t_A = \{\textbf{y}_1(A), \textbf{y}_2(A), \ldots, \textbf{y}_t(A)\}$ where $A \subseteq \{1,2, \ldots, n\}$. Denote the parameters associated with $\textbf{Y}_t$ by $\pmb{\theta}_t = \{\theta_t(1), \theta_t(2), \ldots, \theta_t(n)\}$ such that $\theta_t(i)$ is the parameter vector associated with component $Y_t(i)$. 

%Further, let $\textbf{y}^t = \{\textbf{y}_1, \textbf{y}_2, \ldots, \textbf{y}_t\} = \{\textbf{y}_s\}_{s\leq t}$, $\;\textbf{y}^t_i = \{y_1(i), y_2(i), \ldots, y_t(i)\} = \{y_s(i)\}_{s\leq t}$ and $\textbf{y}^t_A = \{\textbf{y}_1(A), \textbf{y}_2(A), \ldots, \textbf{y}_t(A)\}$ where $A \subseteq \{1,2, \ldots, n\}$. Denote the parameters associated with $\textbf{Y}_t$ by $\pmb{\theta}_t = \{\theta_t(1),\ldots, \theta_t(n)\}$ such that $\theta_t(i)$ is the parameter vector associated with component $Y_t(i)$. 

\begin{proposition}
For a dynamic model over a time series $\textbf{Y}_t = \{Y_t(1), Y_t(2), \ldots, Y_t(n)\}$ such that its conditional independence structure can be represented by a DAG whose vertices are the components of $\textbf{Y}_t$ and the prior parameter vectors denoted by $\pmb{\theta}_0$ are set to be mutually independent, then the following conditional independencies hold
\begin{align}
    \indep_{i \in [n]} \theta_t(i) \,&\mid\, \textbf{y}^{t} \label{ci:mdm1} \\
     \theta_t(i) \,\indep\,  \textbf{y}^t_{[n]\backslash \{i \cup Pa(i)\}} \,&\mid\, \textbf{y}^t_i, \textbf{y}^t_{Pa(i)} \label{ci:mdm2}
\end{align}
\noindent where $\indep$ stands for probabilistic independence, $|$ shows conditioning variables on the right, $[n] = \{1, 2, \ldots, n\}$, and $Pa(i)$ are the indices in $[n]$ associated with the parents of component $Y_t(i)$ in the DAG of the model.
\label{prop:mdm_CIs}
\end{proposition}

The conditional independence in Statement \ref{ci:mdm1} indicates that the parameter vectors for the different components remain independent for all time given the present and past observations, and the conditional independence in Statement \ref{ci:mdm2} states that given the present and past observations for component $Y_t(i)$ and its parent components in the DAG of the model, the parameter vector for component $Y_t(i)$ is independent of the rest of the observed data \citep{QueenandSmith93}. These conditional independencies ensure that the parameters associated with each component of the dynamic multivariate model can be updated independently at each time $t$ and remain independent thereafter at each future time $t' > t$. This conditional independence structure along with the specified DAG representation of the components of $\textbf{Y}_t$ enables an MDM to decompose $\textbf{Y}_t$ such that each of its components is a univariate DLM. The proofs for Statements \ref{ci:mdm1} and \ref{ci:mdm2} can be found in \cite{QueenandSmith93}. In particular, we note here that the validity of these statements does not rely on each of the components being decomposed into univariate DLMs. The proof in \citet{QueenandSmith93} is an induction that simply relies on the prior parameter vectors $\pmb{\theta}_0$ being mutually independent, and on an application of the d-separation theorem \citep{verma1990causal} on the DAG representation of the components of $\textbf{Y}_t$. Hence, the above decoupling methodology can be easily transferred to suitable non-MDM settings such as ours as described in Section \ref{sec:idss_for_terrorism}.

%%%%%%%%%%%%%%%%%%%%%%%%%%%%%%%%%%%%%%%%%%%%%%%%%%%%%%%%%%%%%%%%%%
%%%%%%%%%%%%%%%%%%%%%%%% TERRORIST NETWORK %%%%%%%%%%%%%%%%%%%%%%%
%%%%%%%%%%%%%%%%%%%%%%%%%%%%%%%%%%%%%%%%%%%%%%%%%%%%%%%%%%%%%%%%%%

\section{Terrorist Network} \label{sec:dynamic_network}
% Add the notation and the information about the estimation of the edge weights

% Recall that $\calP_t^*$ denotes the open population of POIs at time $t$ and $\calP_t \subseteq \calP_t^*$ denotes the subset of individuals that the authorities have decided to investigate and monitor at time $t$. 
The Terrorist Network is an undirected, dynamic, and weighted network model which, at each time $t \geq 0$ consists of the POIs $\calP_t$ whom the counter-terrorism authorities choose to pursue at time $t$ as the vertices and edges indicate known or potential ties between these POIs. Ties between the suspects are informed by observable data as described in Section \ref{sec:observable_data}. During each time period, new leads are discovered. From among these leads, new investigative cases are opened for those that pass a triage process meeting defined criteria\footnote{These criteria include: (i) Risk: is there any evidence of risk in intelligible form, (ii) Credibility: is the information reliable; (iii) Actionable: can anything actually be done about it; (iv) and Proportionality: is investigation of the lead necessary and proportionate within legal and statutory obligations, resources and priorities \citep{Anderson2016}.}. The triage process thus gives rise to a set of newly identified individuals $\calP_t^{+}$ at each time interval $t$. Over the same interval, a set $\calP_t^{-}$ are lost from $\calP_t$ for a variety of reasons such as death, arrest, evidence of innocence, physical movement to leave the jurisdiction of the authorities, or de-prioritisation based on evidence and case load. For simplicity, assume that $\calP_t^{+}$ join the set $\calP_t$ at the start of the time period $t$ and existing individuals $\calP_t^{-}$ are lost at the end of $t$. We then have that
\begin{equation}
\calP_t = \{\calP_{t-1} \backslash \calP_{t-1}^{-}\} \cup \calP_t^{+}.
\end{equation}
An undirected network $\mathcal{N}_{t} = (V(\calN_t), E(\calN_t))$ is then created at each time $t$ where $V(\calN_t) = \calP_t$ are the vertices and $E(\calN_t)$ are the edges of the network. An edge exists $e_{ij} \in E(\calN_t)$ between two individuals \textit{p}$_i$\ and \textit{p}$_j$\ if there is a tie between them. For example, if
\begin{enumerate}[label=(\alph*)]
    \itemsep0em
    \item They share an existing familial or social link;
    \item They have committed crimes together in the past;
    \item They have shared affiliations;
    \item Since becoming POIs, they have been observed communicating with each other.
\end{enumerate}

Once an edge is created in $\mathcal{N}_t$ between some \textit{p}$_i$, \textit{p}$_j$\ $\in \calP_t$, this edge endures for all $\mathcal{N}_{t'}$ where $t' \geq t$ as long as \textit{p}$_i$, \textit{p}$_j$\ $\in \calP_{t'}$. Denote by $\varphi_{ijt}$ the latent random variable measuring the pairwise communications shared between \textit{p}$_i$\ and \textit{p}$_j$\ at time $t$. Thus $\varphi_{ijt}$ acts as a quantitative measure of the information directly exchanged between \textit{p}$_i$\ and \textit{p}$_j$ and models the edge weight on the edge $e_{ij}$ in $\mathcal{N}_t$. Denote by $\Phi_t$ a $\abs{\calP_t} \times \abs{\calP_t}$ symmetric matrix with its $(i,j)$th entry given by $\varphi_{ijt}$. By convention, we set $\varphi_{ijt} = 0$ if $i = j$ or $e_{ij} \not \in E(\calN_t)$ for $i \neq j$. The observable pairwise communications data (see Section \ref{sec:observable_data}) are used to estimate $\varphi_{ijt}$. Note that the granularity of the time steps, for example hourly, daily or weekly, is chosen to suit the observation process.
%\footnote{Note here that, unlike several other network models where the edge weight is a real number, we model the edge weight here with a random variable whose expectation could be viewed as a scalar edge weight, if required.} 

\begin{example}(A simple hypothesised terrorist group)
\label{ex:criminal_example}
Consider here the investigative activities of authorities in a particular hypothetical town in the UK. The time steps here are assumed to be weekly. Four individuals from this town, namely \textit{p}$_1$,\ \textit{p}$_2$,\ \textit{p}$_3$ and \textit{p}$_4$ have been observed to have posted pro-terrorist material on social media and have been triaged into $\calP_t$, the observed subpopulation at time $t$. Also at time $t$, a separate lead reveals the return of an individual \textit{p}$_5$ -- who was known to previously have pro-terrorist ideas -- from a country whose local terrorist groups are known to run large radicalisation campaigns targeting foreign individuals. Thus, the subpopulation of suspects under investigation at time $t$ is given by $\calP_t =\{$\textit{p}$_1$,\ \textit{p}$_2$,\ \textit{p}$_3$,\ \textit{p}$_4$,\ \textit{p}$_5 \}$. The preliminary investigation revealed that \textit{p}$_1$\ and \textit{p}$_2$\ attended the same secondary school and are the same age, and that \textit{p}$_2$\ and \textit{p}$_3$\ attend the same gym and are frequently seen together. Further, \textit{p}$_4$\ and \textit{p}$_5$\ were both known affiliates of, a now defunct, local criminal group, and \textit{p}$_1$\ and \textit{p}$_5$\ were arrested together for a minor offence in the past. Due to these pre-existing links, edges $e_{1,2}$, $e_{2,3}$, $e_{4,5}$ and $e_{1,5}$ can be created at time $t$, see Figure \ref{fig:example1}(a). In the duration of the week represented by time $t$, \textit{p}$_5$ has been arrested and extradited to another country, with which the UK shares an extradition agreement, on a serious accusation of kidnapping and murder. Thus, \textit{p}$_5$ is no longer a POI to these authorities, and hence, $\calP_{t+1} =\{$\textit{p}$_1$,\ \textit{p}$_2$,\ \textit{p}$_3$,\ \textit{p}$_4\}$. The network at time $t+1$ is represented in Figure \ref{fig:example1}(b). The authorities continue their monitoring activities on these individuals through the weeks represented by time $t+1$ and $t+2$ with no changes to the structure of the network. At time $t+2$ it is discovered that mobile phones newly registered to the addresses of  \textit{p}$_1$ and \textit{p}$_4$ are in communication. This creates a tie between \textit{p}$_1$\ and \textit{p}$_4$\ as shown in Figure \ref{fig:example1}(c). 

\begin{figure*}[ht]  % spans both columns
\centering
\begin{subfigure}{0.30\textwidth}
\includegraphics[scale = 0.35, trim ={1cm, 0, 0, 0}]{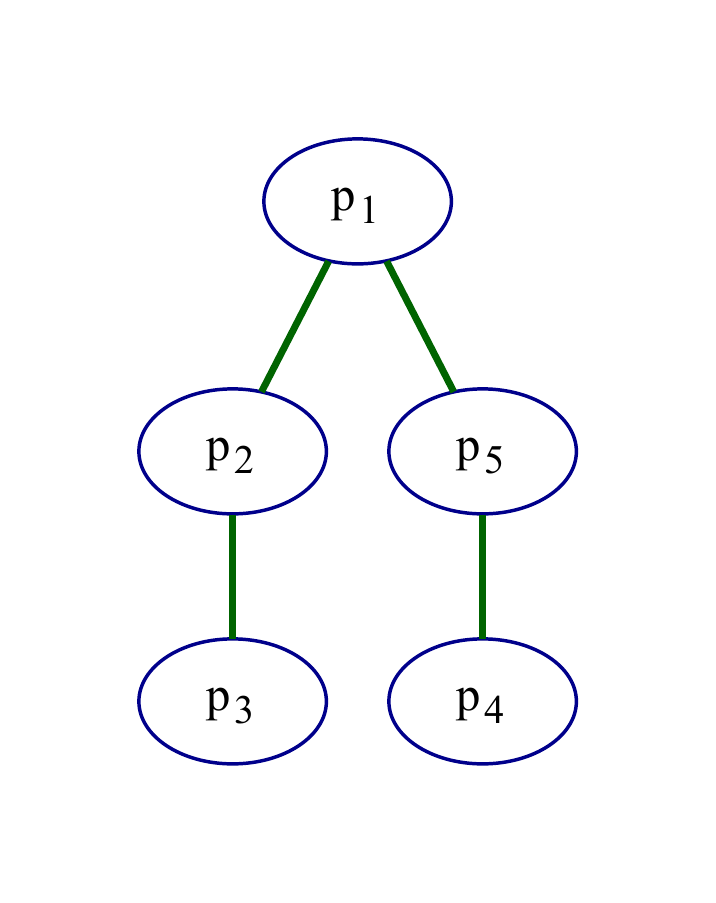}
\caption{At time $t$}
\label{fig:ex1(a)}
\end{subfigure}
\begin{subfigure}{0.30\textwidth}
\includegraphics[scale = 0.35, trim ={0.8cm, 0, 0, 0}]{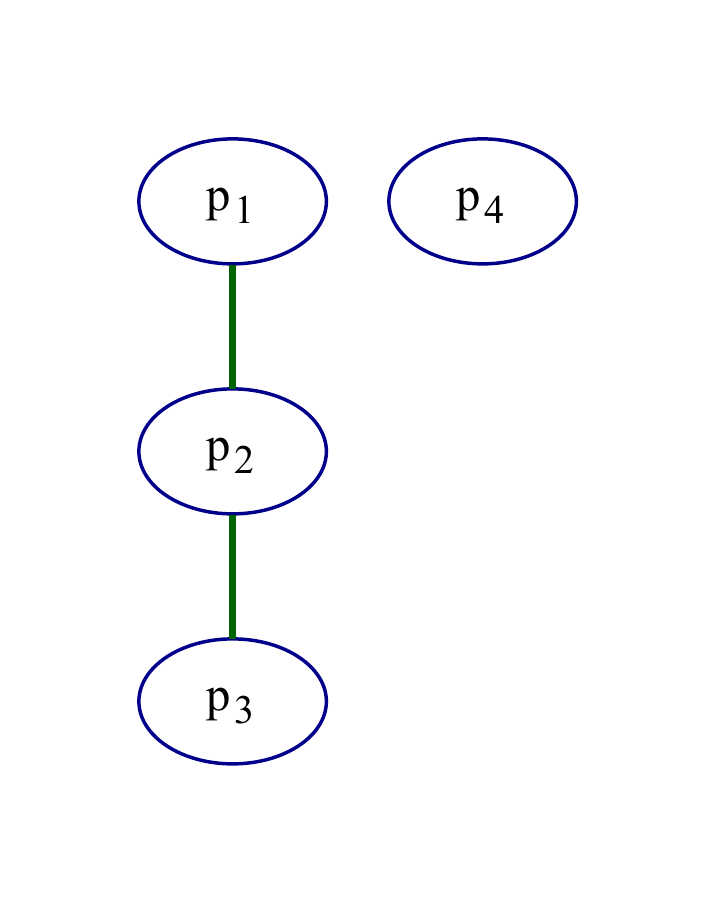}
\caption{At time $t+1$ and $t+2$}
\label{fig:ex1(b)}
\end{subfigure}
\begin{subfigure}{0.30\textwidth}
\includegraphics[scale = 0.35, trim ={1cm, 0, 0, 0}]{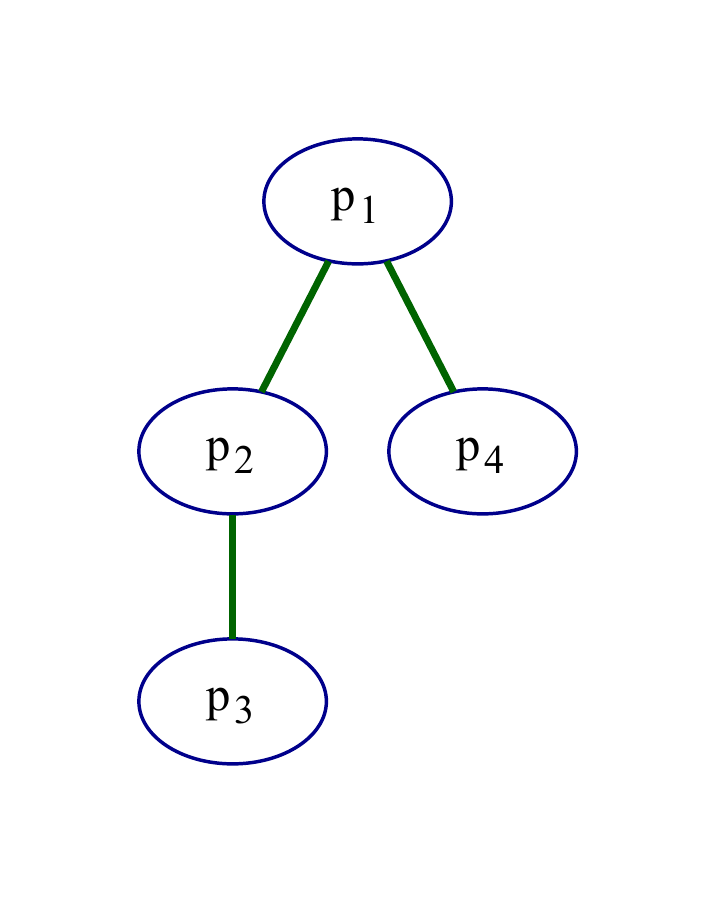}
\caption{At time $t+3$}
\label{fig:ex1(d)}
\end{subfigure}

\caption{Structure of network at times $t$, $t+1$, $t+2$ and $t+3$.} % Overall figure caption
\label{fig:example1}
\end{figure*}

\end{example}

Recall that the counter-terrorism authorities are likely to receive data and information from multiple sources. Suppose that there are $K$ such \textit{information channels}. The data from each channel is condensed into a summary measure in the Terrorist Network. The summary measure used for each channel depends on factors such as the type of data, the data source, the frequency of the observations, and the required granularity of that specific type of data. For instance, for an information channel informing the duration of phone calls or number of text messages exchanged between a pair of suspects, the summary measure may simply be the sum of the observations, whereas for the amount of money exchanged between the suspects, a first-order difference in the observations might be a suitable summary measure. However, note that these summary measures for the different information channels may be on very different scales of measurement, e.g. $x$ hours of a phone call and \pounds $x$ of money exchanged, and hence, might have an disproportionate effect on the edge weight variable $\varphi_{ijt}$. To balance the effect of data relating to different channels on $\varphi_{ijt}$, the data obtained through the different channels must be on a comparable scale. This can be achieved through any of the standard methods of scaling or normalisation (see e.g. \citet{jahan2015state}).

Denote by $s_{ijkt}$ the scaled or normalised summary measure of the data observed between the pair \textit{p}$_i$, \textit{p}$_j$\  $\in \calP_t$ from channel $k$ at time $t$. We assume that the following independence relationship holds:
\begin{equation}
    \indep_{k \in \{1,\ldots,K\}} \, s_{ijkt}
    \label{eq:info_independence}
\end{equation}
\noindent which implies that the data and information obtained from the different information channels for a given pair $\{\text{\textit{p}$_i$, \textit{p}$_j$}\} \in \calP_t \times \calP_t$ at time $t$ are mutually independent. Denote by $S_t$ the \textit{observations matrix} at time $t$ with elements $\textbf{s}_{ijt}$ such that $\textbf{s}_{ijt} = \{s_{ij1t}, \ldots, s_{ijKt}\}$. Notice that $S_t$ is a symmetric $\abs{\calP_t} \times \abs{\calP_t}$ matrix with $\textbf{s}_{ijt} = \textbf{s}_{jit}$ due to the nature of the pairwise communications data. We use the convention that $\textbf{s}_{ijt}$ is a $K$-dimensional zero vector whenever $i = j$, $e_{ij} \not \in E(\calN_t)$ for $i \neq j$, and whenever no information is observed between two individuals. To indicate the difference in the quality or reliability of data obtained from the different channels, we define a parameter $\xi_k \in (0,1]$ which denotes the \textit{efficiency} of the intelligence obtained from channel $k$, for $k = 1, \ldots, K$. This efficiency parameter indicates the loss of information expected from a specific information channel. A value closer to 1 represents minimal loss of information (e.g. bank transitions data), whereas a value closer to 0 indicates that the actual observations are likely to be much higher than what has been conveyed to the authorities (e.g. patchy or poor source of secondary data).

\begin{example}(A simple hypothesised terrorist group (continued))
Consider pairs $\{\text{\textit{p}}_1, \text{\textit{p}}_2\}$ and $\{\text{\textit{p}}_2, \text{\textit{p}}_3\}$. Suppose that, for each pair, the authorities are able to observe the number of hours of telephone conversations (channel $k = 1$), and the number of text messages exchanged (channel $k = 2$). Further, suppose that for the pair $\{\text{\textit{p}}_2, \text{\textit{p}}_3\}$, after acquiring the required permissions, the authorities also have access to the amount of money exchanged between the pair through bank transactions (channel $k = 3$). Whilst phone records are known to be imperfect, the bank transactions data is considered to be very reliable. Thus, the efficiency parameters are set as $\xi_1 = 0.8, \xi_2 = 0.8$ and $\xi_3 = 1$. 

The data observed from the different information channels are on very different scales. Hence, the authorities decide to scale them to be between 0-10 to balance their effect on $\varphi_{ijt}$. The scaling might be performed as follows
\begin{equation*}
    s_{ijkt} = \frac{r_{ijkt}}{r_{max_k}} \times 10,
\end{equation*}
where $r_{ijkt}$ is the raw data associated with channel $k$ at time $t$ for the pair $\{\text{\textit{p}}_i, \text{\textit{p}}_j\}$ and $r_{max_k}$ is the maximum value expected to be observed for information coming through channel $k$. Typically, $r_{ijkt} \leq r_{max_k}$. If $r_{ijkt} > r_{max_k}$, then either $s_{ijkt} > 10$ may be used in the updating for $\varphi_{ijt}$, or $s_{ijkt}$ can be set to $10$ for all $r_{ijkt} > r_{max_k}$ when the effect of increases beyond $r_{max_k}$ are not considered to have substantial marginal effects. Table \ref{tab:multiple_channels} shows the scaled and raw values of the information observed from the different information channels for both pairs over times $t_0, t_1$ and $t_2$. Here, $r_{max_k}$ is set to 35, 1400 and 70000 for $k = 1, 2, 3$ respectively. 

\begin{table}[bt]
\caption{An example of scaled (raw) summary measures of the data observed between pairs $\{\text{\textit{p}}_1, \text{\textit{p}}_2\}$ and $\{\text{\textit{p}}_2, \text{\textit{p}}_3\}$. Here $s_{\cdot, \cdot, k, \cdot}$ refers to data related to telephone conversations for $k = 1$, text messages exchanged for $k = 2$, and money exchanged through bank transfers for $k = 3$.}
\label{tab:multiple_channels}
\begin{threeparttable}
\begin{tabular}{l|ll|ll|cl}
\headrow
\thead{Time} & \thead{$s_{1,2,1, \cdot}$} & \thead{$s_{2,3,1, \cdot}$} & \thead{$s_{1,2,2, \cdot}$}
& \thead{$s_{2,3,2, \cdot}$} & \thead{$s_{1,2,3, \cdot}$} & \thead{$s_{2,3,3, \cdot}$} \\
\hline
\hiderowcolors
$t_0$ & 0.857 (3) & 2.286 (8) & 0.357 (50) & 0.107 (15) & -- & 0.071 (500) \\
$t_1$ & 0.571 (2) & 3.429 (12) & 1.786 (250) & 0.1429 (20) & -- & 0.014 (100) \\
$t_2$ & 1.143 (4) & 4.286 (15) & 1.25 (175) & 0.1429 (20) & -- & 0.4 (2800) \\
\hline
\end{tabular}
\end{threeparttable}
\end{table}

\end{example}

In order to maintain transparency in the model, interpretability of its parameters, and to enable quick and efficient inference, we use a Gamma-Poisson conjugate setting for updating the distributions of the of $\varphi_{ijt}$ -- the random variables modelling the edge weights, for \textit{p}$_i$, \textit{p}$_j$\ $\in \calP_t$ and $t \geq 0$. To facilitate this, we adopt the approach of using discount factors to transform the posterior at time $t$ into the prior at time $t+1$ as described in \citet{WestandHarrison,Smith79}. The discount factor is a value in $(0, 1]$ that represents the decay of information from time $t-1$ to time $t$. Additionally, we assume the following plausible conditional independence relationships 
  \begin{align}
        \varphi_{ijt} &\indep \mathcal{F}_{t} %\backslash \{\varphi_{ij,t-1}\} 
        \,\mid\, \varphi_{ij,t-1}, \label{eq:markov_assumption} \\
        s_{ijt} &\indep (\Phi_{t}%\backslash \{\varphi_{ijt}\}
        , S_t %\backslash \{s_{ijt}\}
        , \mathcal{F}_{t}) \, \mid\, \varphi_{ijt}. \label{eq:output_independence}
  \end{align}
\noindent where $\mathcal{F}_{t}$ denotes all past data and edge weight random variables up to but not including time $t$, i.e. $S_{t'}$ and $\Phi_{t'}$ for ${t' < t}$. Statement \ref{eq:markov_assumption} is a standard first-order Markov assumption and Statement \ref{eq:output_independence} implies that the pairwise communications data between any pair of suspects $\{\text{\textit{p}$_i$, \textit{p}$_j$}\} \in \calP_t \times \calP_t$ at a given time $t$ are only dependent on $\varphi_{ijt}$, the quantitative measure of information being exchanged between them at that time. This enables us to update the edge weight variables $\varphi_{ijt}$ using observational data $\textbf{s}_{ijt}$ for each pair \textit{p}$_i$ and \textit{p}$_j$\ independently, see Appendix B.

We now describe the forward filtering equations for each pair $\{\text{\textit{p}$_i$, \textit{p}$_j$}\} \in \calP_t \times \calP_t$ in the Terrorist Network:

\subparagraph{Initialisation:} Set the prior $\varphi_{ijt_0}$ as follows
\begin{equation}
\varphi_{ijt_0} \sim Gamma(\alpha_{ijt_0},\beta_{ijt_0})
\label{eq:gamma}
\end{equation}
\noindent where $t_0$ is the first time step of the time-series. The parameters $\alpha_{ijt_0}$ and $\beta_{ijt_0}$ are determined by existing case knowledge. For example, if $e_{ij} \in E(\calN_{t_0})$ exists only due to a social relation, then $\alpha_{ijt_0}$ and $\beta_{ijt_0}$ may be set such that the mean and variance of $\varphi_{ijt_0}$ are both relatively low. On the other hand, if $\calP_i$ and $\calP_j$ have a previous joint conviction then these parameters can be set such that the $\varphi_{ijt_0}$ has a high mean and lower variance. 
\subparagraph{Posterior at time $t-1$:} Let the posterior of $\varphi_{ij,t-1}$ after observing $s_{ij,t-1}$ and $\mathcal{F}_{t-1}$ be given by
\begin{equation}
\varphi_{ij,t-1} \,\mid\, \textbf{s}_{ij,t-1}, \mathcal{F}_{t-1} \sim Gamma(\alpha_{ij,t-1},\beta_{ij,t-1}).
\label{eq:post_at_t-1}
\end{equation}
\subparagraph{Prior at time $t$:} Using the discount factor $\delta_{ijt} \in (0,1]$, the posterior at time $t-1$ evolves to the prior at time $t$ as 
\begin{equation}
\varphi_{ijt} \,\mid\, \mathcal{F}_{t} \sim Gamma(\delta_{ijt} \alpha_{ij,t-1},\delta_{ijt} \beta_{ij,t-1}).
\end{equation}
Under this posterior-to-prior evolution, the mean of the distribution remains unaffected while the variance either remains the same (when $\delta_{ijt} = 1$) or increases (when $0 < \delta_{ijt} < 1$). Thus, a lower value of $\delta_{ijt}$ indicates a reduced confidence in the posterior at the previous time step as the variance increases. This is also associated with a decay of information from the previous time step depending on how much the situation is likely to have evolved since then. 

%[MAYBE ADD TO SUPPLEMENTARY MATERIALS]
% \begin{align}
%     \mathbb{E}[\varphi_{ijt} \,\mid\, \mathcal{F}_{t}] &= \dfrac{\delta_{ijt} \alpha_{ij,t-1}}{\delta_{ijt} \beta_{ij,t-1}} \nonumber \\
%     &= \dfrac{ \alpha_{ij,t-1}}{ \beta_{ij,t-1}} \nonumber \\
%     &= \mathbb{E}[\varphi_{ij,t-1} \,\mid\, s_{ij,t-1}, \mathcal{F}_{t-1}]
% \end{align}
% %
% \begin{align}
%     \text{var}[\varphi_{ijt} \,\mid\, \mathcal{F}_{t}] &= \dfrac{\delta_{ijt} \alpha_{ij,t-1}}{(\delta_{ijt} \beta_{ij,t-1})^2} \nonumber \\
%      &= \dfrac{\alpha_{ij,t-1}}{\delta_{ijt} (\beta_{ij,t-1})^2} \nonumber \\
%      &\geq \dfrac{\alpha_{ij,t-1}}{\beta_{ij,t-1}^2} \nonumber \\
%      &= \text{var}[\varphi_{ij,t-1} \,\mid\, s_{ij,t-1}, \mathcal{F}_{t-1}]
% \end{align}
% %
% \noindent Thus, the mean remains constant in the posterior to prior transformation and hence, this is known as a ``steady" or stable evolution. The variance, however, either remains the same (when $\delta_{ijt} = 1$) or increases (when $0 < \delta_{ijt} < 1$). Thus, when the value of $\delta_{ijt}$ is close to one, then our prior for $\varphi_{ijt}$ at time $t$ is approximately identical to the posterior for $\varphi_{ij,t-1}$. A lower value of $\delta_{ijt}$ indicates a reduced confidence in the posterior at the previous time step as the variance increases. This is also associated with a decay of information from the previous time step depending on how much the situation is likely to have evolved since then. 

\subparagraph{Data generation at time $t$:} The observations from the different information channels are modelled independently as
\begin{equation}
s_{ijkt} \,\mid\, \varphi_{ijt}, \mathcal{F}_{t} \sim Poisson(\xi_k\varphi_{ijt}), \quad k = 1, \ldots, K.
\end{equation}
\subparagraph{Posterior at time $t$:} The posterior when the observation vector $\textbf{s}_{ijt}$ has at least one non-zero element is given by 
\begin{align}
p(\varphi_{ijt} \,\mid\, \textbf{s}_{ijt}, \mathcal{F}_{t}) &\propto 
\notag \prod_{k=1}^K p(s_{ijkt}\,\mid\,\varphi_{ijt}, \mathcal{F}_{t}) \, p(\varphi_{ijt} \,\mid\,\mathcal{F}_{t}) \notag \\
&=    \varphi_{ijt}^{\sum_k s_{ijkt} + \delta_{ijt}\alpha_{ij,t-1} -1} \, \exp(-(\textstyle \sum_k \xi_k + \delta_{ijt}\beta_{ij,t-1})\varphi_{ijt}) 
\notag \\
\varphi_{ijt} \,\mid\, \textbf{s}_{ijt}, \mathcal{F}_{t} &\sim Gamma(\alpha_{ijt},\beta_{ijt})
\end{align}
\noindent where $\alpha_{ijt} = \delta_{ijt} \alpha_{ij,t-1} + \sum_k s_{ijkt}$ and $\beta_{ijt} = \delta_{ijt} \beta_{ij,t-1} + \sum_k \xi_k$. For the same value of $\sum_k s_{ijkt}$, a lower overall efficiency of the observations given by $\sum_k \xi_k$ results in a higher mean and larger variance -- indicating the associated increase in uncertainty -- of $\varphi_{ijt}$ compared to when the overall efficiency is higher. 

%[DO we need this???]
% \subparagraph{One-step-ahead forecast:} The one-step-ahead forecast of the data from channel $k$ can be obtained in closed form as
% \begin{align}
%     p(s_{ijk,t+1} \,\mid\, \mathcal{F}_{t+1}^-) &= \int_{\varphi_{ij,t+1}} p(s_{ijk,t+1}\,\mid\,\varphi_{ij,t+1}, \mathcal{F}_{t+1}^-) \, p(\varphi_{ij,t+1}\,\mid\, \mathcal{F}_{t+1}^-) \, d \varphi_{ij,t+1}   \notag\\
%      &= \binom{s_{ijk,t+1}+ \delta_{ij,t+1} \alpha_{ijt} - 1}{s_{ijk,t+1}} \dfrac{(\delta_{ij,t+1}\beta_{ijt})^{\delta_{ij,t+1}\alpha_{ijt}}\, \xi_k^{s_{ijk,t+1}}}{ (\xi_k + \delta_{ij,t+1}\beta_{ijt})^{(\delta_{ij,t+1}\alpha_{ijt} + s_{ijk,t+1})}}
% \end{align}
% \noindent where $\binom{\boldsymbol{\cdot}}{\boldsymbol{\cdot}}$ denotes the binomial coefficient.

The distribution of $\varphi_{ijt}$ for a pair $\{\text{\textit{p}$_i$, \textit{p}$_j$}\}$ can hence be periodically updated over the evolution of time $t$ in closed-form using the above recurrences across the Terrorist Network given sequential incoming observational data. The dynamic nature of the open population is easily incorporated in our model by introducing vertices, edges and priors for immigrants (new entrants) and removing them for emigrants (leavers) at the appropriate time. Finally, we note here that in a policing and counter-terrorism setting, it is essential to differentiate between the following cases:

\begin{enumerate}
    \itemsep0em
    \item $\sum_k s_{ijkt} = 0$ because \textit{p}$_i$\ and \textit{p}$_j$\ were monitored but did not communicate in any way during time $t$;
    \item $\sum_k s_{ijkt} = 0$ because \textit{p}$_i$\ and \textit{p}$_j$\ were not closely monitored during time $t$.
\end{enumerate}

In the first case, the posterior update is carried out as described above as we have \textit{observed} zero communications. Notice that if no new information is observed through $\textbf{s}_{ijs}$, $s \geq t$ then the variance of $\varphi_{ijs}$, $s \geq t$ will keep increasing. To prevent this and to reflect that we expect a baseline amount of information flow to continue between a pair of suspects \textit{p}$_i$\ and \textit{p}$_j$\ who share an edge between them -- until we observe information indicating otherwise -- we can set the discount factor as $\delta_{ijt} =  d_{ij} + (1 -d_{ij})  \exp(- \sum_k s_{ijk,t-1}\xi_{k})$ as detailed in \citet{Chenwest}. Here $d_{ij}$ is the baseline discount factor for pair $\{\text{\textit{p}$_i$, \textit{p}$_j$}\}$. This is particularly useful if we expect to have large consecutive gaps of time when we do not expect to observe good quality data on the pairs. When we observe very low levels of quality information in the previous time, the discount factor is closer to 1 and when good quality information is observed, the discount factor will be closer to $d_{ij}$. This setting allows us to set pair-specific discount factors if required. 

\begin{example}(A simple hypothesised terrorist group (continued)) Table \ref{tab:parameters_multiple} shows the prior and posterior parameters for $\varphi_{\cdot, \cdot, \cdot}$ for pairs $\{\text{\textit{p}}_1, \text{\textit{p}}_2\}$ and $\{\text{\textit{p}}_2, \text{\textit{p}}_3\}$ at times $t, t+1$ and $t+2$. For both pairs, the prior at time $t$ is set to have a low mean with some uncertainty. The updating is carried out as described above with discount factor as $\delta_{\cdot, \cdot, \cdot} = 0.7$ for all $t$.

\begin{table}
\caption{Evolution of the prior and posterior parameters for $\varphi_{\cdot, \cdot, \cdot}$ for time periods $t, t+1$ and $t+2$.}
\label{tab:parameters_multiple}
\begin{threeparttable}
\begin{tabular}{c|llll l c|llll}
\headrow
\thead{Time} &   \thead{$\varphi_{1,2, \cdot}$} &  &   \thead{$\varphi_{2,3, \cdot}$} & & & \thead{Time} &   \thead{$\varphi_{1,2, \cdot}$} &  &   \thead{$\varphi_{2,3, \cdot}$} & \\
&  \thead{$\alpha$} &  \thead{$\beta$} &  \thead{$\alpha$} &   \thead{$\beta$} & & &  \thead{$\alpha$} &  \thead{$\beta$} &  \thead{$\alpha$} &   \thead{$\beta$}\\
\hline
\hiderowcolors
$t$ prior  &   0.70 &       1.41 &    0.70  & 1.41 & & $t+1$ post   &   3.6968 &   3.707 & 5.8007  &  5.407\\
$t$ post   &   1.914 &       3.01 &  3.164   &  4.01 & & $t+3$ prior  &   2.5878 &      2.5949 & 4.0605    &   3.7849 \\
$t+1$ prior  &   1.3398 &     2.107 & 2.2148    &  2.807 & & $t+3$ post   & 4.9808 & 4.1949 &  8.8894 & 6.3849 \\
\hline
\end{tabular}
\end{threeparttable}
\end{table}

\end{example}

% To summarise, Figure \ref{fig:collaboration_model} gives an overview of the dynamic network model.

% \begin{figure}[h]
% \vspace{1in}
% \centering
% \includegraphics[trim = 1cm 0cm 0cm 6.7cm, scale = 0.53 ]{img/criminal_collaboration.pdf}
% \caption{Overview of the dynamic network model.}
% \label{fig:collaboration_model}
% \end{figure}

\subsection{Related Work} \label{subsec:related_work}

% Add here details of other related network modelling approaches generally to social networks and in particular to terrorist networks

%\textcolor{red}{[ADITI: Could this section be shortened perhaps?}

Network data relating to activities of opposition and terrorist forces have been analysed using SNA techniques including link analysis since at least World War II \citep{FTA1948,Meter2002}. Examples of link analysis and network survey methods being used to gain intelligence and strategise can also be found from The Troubles in Northern Ireland where these methods were used to ``target individuals to assassinate", and in Thailand where the US army used "village survey" methods to ``interview village members and note family and community relationships" in connection to the fight against the Communist rebels \citep{Meter2002}.

Within academic contexts, the merits of network analysis for terrorism research were originally assessed by \cite{SPARROW1991} in his seminal paper. The author motivated the importance of network analysis concepts such as centrality, node degree, betweenness, closeness, stochastic equivalence and Euclidean centrality after multidimensional scaling within the context of criminal and terrorist networks. He further emphasised issues such as ``weak ties" which indicate that the most valuable and urgent communication channels are likely to be those ``which are seldom used and which lie outside the relatively dense clique structures", ``fuzzy boundaries" which indicate that boundaries of such networks can be quite ambiguous, and ``incompleteness" indicating that data relating to these networks are likely to be incomplete with informative missingness. Prior to the work of \cite{SPARROW1991}, the leading method of network analysis within law enforcement was the Anacapa charting system \citep{harper1975application} developed by Anacapa Sciences Inc., California and widely used since its introduction. This charting system provided a two-dimensional visual representation of the link data and allowed the user to clearly pick out features such as links, centrality, cliques etc but the charting system itself did not involve any analysis of the data these charts represented. 

Following \cite{SPARROW1991}, the application of SNA methods within criminal and terrorist networks has been researched extensively. This includes using centrality measures including node removal methods to identify key individuals and heterogeneous roles \citep{Toth2013, lee2012criminal, memon2006practical, berzinji2012detecting}, bipartite and multipartite graph methods to identify overlapping cells \citep{ranciati2017, Campedelli2019}, dynamic line graphs to visualise the temporal dynamics of terrorist actors in covert actions and events \citep{Broccatelli2016}, spectral clustering methods to identify criminal groups \citep{van2013community}, graph distances to analyse criminal networks in presence of missing data \citep{ficara2021criminal}, link prediction methods within criminal and terrorist networks \citep{lim2019hidden, budur2015structural, berlusconi2016link, rhodes2007social}, and community detection in these networks using SNA methods \citep{ferrara2014detecting, xu2005crimenet, robinson2018detection, bahulkar2018community}. SNA methods have been used to analyse terrorist networks of varying levels of threat such as the terrorist networks of the September 11, 2001 terrorists \citep{Krebs2002mapping}, al-Qaeda \citep{sageman2004understanding}, the terrorists who carried out the 2015 and 2016 Islamic State attacks at Paris and Brussels \citep{remmers2019temporal}, the Salafi Jihad Network \citep{qin2005analyzing}, and the Sicilian Mafia \citep{calderoni2020robust}. However, these methods typically only cater to a fixed data source and do not take into account the criminal trajectory of each individual member of a terrorist group. 

Our proposed two component IDSS is closest to the approach taken by DARPA, Aptima Inc. and the University of Connecticut who developed the Adaptive Safety Analysis and Monitoring (ASAM) tool (for details, see \citet{allanach2004detecting, singh2004stochastic}). ASAM is a hierarchical system consisting of a collection of hidden Markov models (HMMs) at its lowest level where each HMM monitors the probability of a specific type of terrorist attack. The outputs from the HMMs feed into a series of subordinate dynamic Bayesian networks (DBNs) -- each subordinate DBN represents a different terrorist model -- which, in turn, feeds into a larger DBN that evaluates the overall probability of a terrorist attack. The ASAM approach is similar to our proposed IDSS, in that it extracts signals of suspicious activities from noisy, partial data based on the understanding that any terrorist attack has a goal for which certain tasks must be carried out and certain communications must take place among a terrorist network for the goal to be achieved. However, unlike the ASAM approach, our two component IDSS approach is Bayesian which allows critical expert knowledge and instincts to be incorporated into the priors for the IDSS and all its parameters can be analytically derived which imparts complete transparency to the working of the IDSS. Additionally, our approach takes into account the threat posed by each member of a terrorist group which the ASAM approach does not fully consider. 

% - latent class model to investigate group structure of terrorist groups (https://rss.onlinelibrary.wiley.com/doi/abs/10.1111/rssa.12233?__cf_chl_jschl_tk__=pmd_jQL9gpVPiMBQOCgtzgStj55asSQGPGPtSbmajwiQw1s-1635259699-0-gqNtZGzNAiWjcnBszQyR)
% - stochastic block models for criminal networks (https://arxiv.org/abs/2007.08569)

\section{IDSS for Activities of Terrorist Groups} \label{sec:idss_for_terrorism}

%Jump right into the application

The decoupling methodology of MDMs described in Section \ref{subsubsec:Decoupling} enables us to formally decouple the Terrorist Network and the individual Lone Terrorist models of each \textp $\in \calP_t$ for each time $t \geq 0$ and then recombine them within a modular IDSS. In fact, this decoupling methodology, although well-demonstrated within MDMs, has not been exploited out of this setting until now. Recall that the properties of this methodology rely only on the initial independencies set through the prior parameters and on the DAG structure linking the components of the time-series. 

We first briefly review the required notation. $\textbf{Y}_t$ refers to the data relating to the activities of a suspect \textit{p}\ at time $t \geq 0$ in their Lone Terrorist model. To generalise this notation to a population of suspects $\calP_t$, let $\textbf{Y}_{it}$ denote the data relating to the activities of suspect \textit{p}$_i$ $\in \calP_t$ at time $t \geq 0$. Further, $\textbf{s}_{ijt}$ is the $K$-dimensional vector containing summary measures of the information shared between individuals \textit{p}$_i$ and \textit{p}$_j$ through the $K$ information channels at time $t \geq 0$. The Terrorist Network model $\calN_t$ for population $\calP_t$ can now be coupled with the $\abs{\calP_t}$ Lone Terrorist models -- one for each \textp $\in \calP_t$ -- through a DAG which contains edges from $\textbf{Y}_{it}$ and $\textbf{Y}_{jt}$ to $\textbf{s}_{ijt}$ for each pair $\{\text{\textit{p}$_i$, \textit{p}$_j$}\} \in \calP_t \times \calP_t$, and no other edges. Recall that $\textbf{s}_{ijt} = \textbf{s}_{jit}$ by design and so we explicitly model only one of these within the DAG. For instance, consider $\calP_t = \{\text{\textit{p}$_i$, \textit{p}$_j$, \textit{p}}_k\}$. The DAG combining the individual Lone Terrorist models for \textit{p}$_i$, \textit{p}$_j$ and \textit{p}$_k$, and the Terrorist Network model $\calN_t$ at time $t \geq 0$ is given in Figure \ref{fig:dag_criminal_collab}. This DAG gives us the structure of the IDSS. 
%\textcolor{blue}{Need to change $u_{ikt}$ to $s_{ikt}$ for consistency with my FOLB changes}
\begin{figure}[h]
\vspace{1in}
\centering
\includegraphics[trim = 1.5cm 0cm 0cm 11.5cm, scale = 0.24 ]{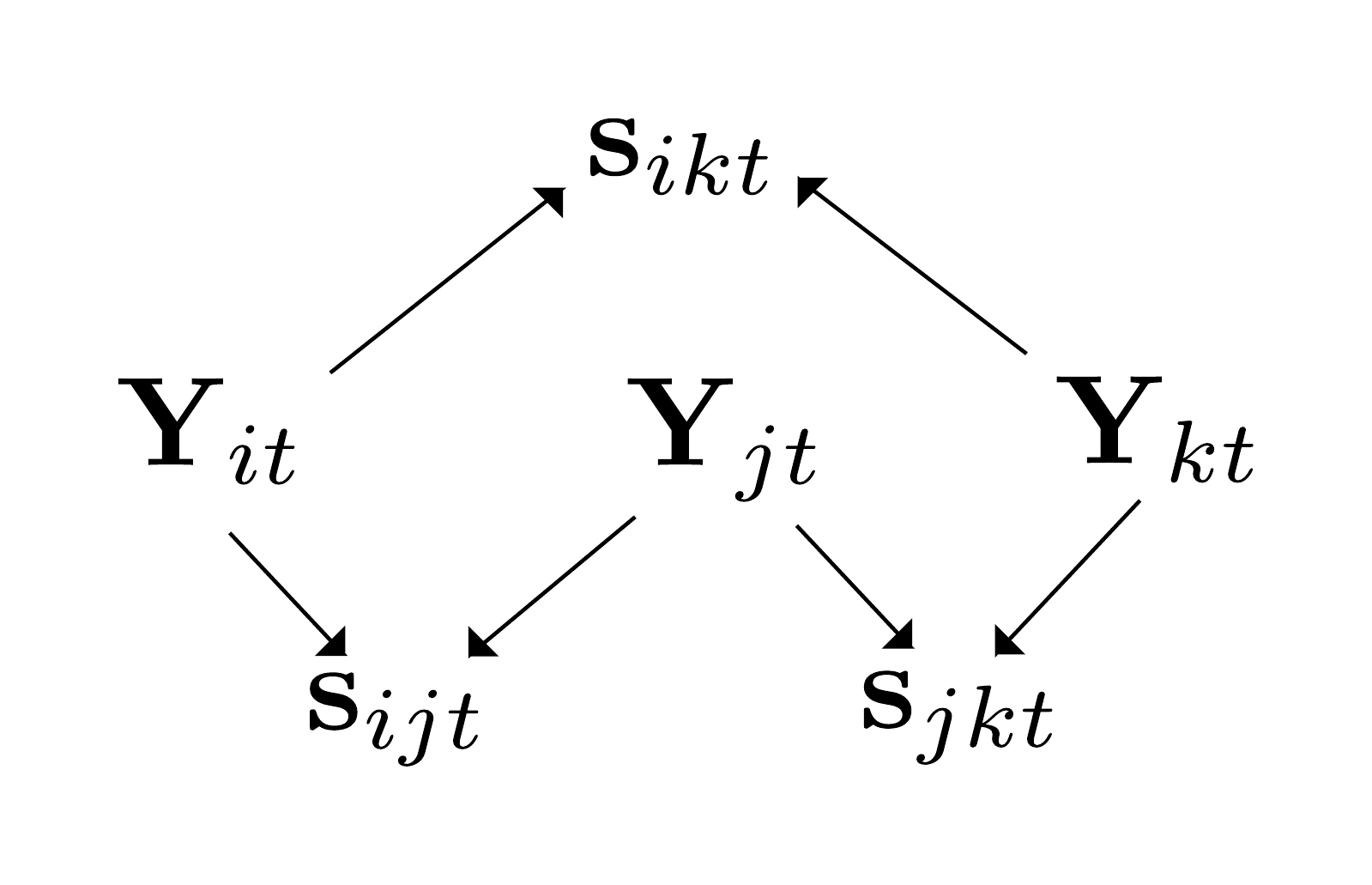}
\caption{The DAG of the IDSS after combining the Lone Terrorist and Terrorist Network models.}
\label{fig:dag_criminal_collab}
\end{figure}

Since $\textbf{s}_{ijt}$ contains all the observed information about the pairwise communications needed to estimate the edge weight modelled by random variable $\varphi_{ijt}$ in the Terrorist Network, and typically $\textbf{s}_{ijt} \subset \textbf{Y}_{it}$ and $\textbf{s}_{ijt} \subset \textbf{Y}_{jt}$, the estimation of $\varphi_{ijt}$ can be performed independently of $\textbf{Y}_{it}$ and $\textbf{Y}_{jt}$ when $\textbf{s}_{ijt}$ is given.

% By adapting Statement \ref{eq:output_independence} as follows
% %
%   \begin{align}
%         s_{ijt} &\indep (\Phi_{t}%\backslash \{\varphi_{ijt}\}
%         , s_t, \{Y_{is}, Y_{js}\}_{s \leq t} %\backslash \{s_{ijt}\}
%         , \mathcal{F}_{t}) \, \mid\, \varphi_{ijt}, \label{eq:output_independence2}
%   \end{align}
% \noindent the one-step-ahead forecasts of $s_{ijt}$ as a product of the one-step-ahead forecasts of its $s_{ijkt}$ components remain independent of $Y_{it}$ and $Y_{jt}$ when $\varphi_{ijt}$ is given. This then allows estimation and forecasts for the RVE models and the networks models to be performed completely independently and then combined together to form the joint model. Note that Statement \ref{eq:output_independence} is a simplifying assumption and can be adapted if necessary.

\subsection{Indicators of a Terrorist Attack} \label{subsec:indicators}

The IDSS combining the Lone Terrorist models and the Terrorist Network for the population of suspects under consideration $\calP_t$ at time $t \geq 0$ can now be used to construct simple yet powerful early warning indicators of the imminence of threat posed by a known or suspected terrorist group. These early warning indicators can be designed to facilitate pre-emptive action to frustrate potential attacks. We demonstrate below how such indicators might be constructed and how they could be utilised to forewarn the counter-terrorism authorities of potential attacks. Note that such indicators would not be able to definitively guarantee that an attack is imminent. Instead, these are meant for authorities to flag the activities of the concerned individuals for increased monitoring and scrutiny. 

We shall, hereafter, refer to suspected or known terrorist groups as \textit{cells}. There are various definitions of what constitutes a terrorist cell. \citet{shapiro2005organizing} states that ``a cell is best understood as an individual or group of individuals that take consequential actions". For our purposes we present a simple definition of a cell within our model as described below. A cell $C \subset \calP_t$ is defined as a group of individuals who induce a connected subgraph in the Terrorist Network $\calN_t$ at time $t \geq 0$. The composition of $C$ is typically determined by the counter-terrorism authorities, based on their expert judgement, such that $C$ is potentially an organisational unit for potential attack. 

We first extend the notation as follows for the Lone Terrorist model of a suspect \textit{p}$_i$ $\in \calP_t$
\begin{align*}
    \pmb{\Theta}_t &= \{\pmb{\vartheta}_{i t} \,:\, \text{\textit{p}}_i \in \calP_t \} \in \{0,1\}^{R \times \abs{\calP_t}}\\
    \pmb{Z}_t &= \{\pmb{Z}_{i t} \,:\, \text{\textit{p}}_i \in \calP_t \} \in \mathbb{R}^{R \times \abs{\calP_t}} \\
    \pmb{X}_t &= \{X_{i t} \,:\, \text{\textit{p}}_i \in \calP_t \} \in
    \{x_0, x_1, \ldots, x_n\}^{\abs{\calP_t}}
\end{align*}
\noindent where $\pmb{\vartheta}_{it}, \pmb{Z}_{it}$ and $\pmb{X}_{it}$ are defined as $\pmb{\vartheta}_{t}, \pmb{Z}_{t}$ and $X_{t}$ respectively in Section \ref{subsec:lone_terrorist}.

%Let $X_t^C$denotes the latent random variable indicating the threat state occupied by the cell $C$ at time $t >0$, where $X_t^C \in \{x_{0}, x_1, \ldots, x_{n}\}$. The data at the surface level $\textbf{Y}_t^C = \cup_{\text{\textit{p}} \in C} \textbf{Y}_t^{\text{\textit{p}}}$ is passed through a filter function to obtain $\textbf{Z}_t^C$ which in turn informs the cell's engagement with the tasks $\pmb{\vartheta}_t^C$ associated with the terrorist attack. Denote by $\pmb{\pi}_t^C = \{\pi_{t0}^C,\pi_{t1}^C, \ldots, \pi_{tn}^C\}$ the probability vector associated with $X_t^C$, such that $\pi_{ti}^C$ indicates the probability that the cell is in state $x_i$ at time $t >0$. These can be obtained through the recurrences described in Appendix [ADD CROSS-REF]. 

\subparagraph{Collective progress:} We can construct a \textit{Terrorist Cell} model for modelling the progress of a cell $C$, as a separate entity, towards a terrorist attack by performing a multivariate extension of the Lone Terrorist model. Here $X_t^C, Z^C_t$ and $\pmb{\pi}_{t}^C$ are defined as $X_t, Z_t$ and $\pmb{\pi}_t$ in Section \ref{subsec:lone_terrorist}. The data $\textbf{Y}_t^C$ at the surface level is given as $\textbf{Y}_t^C = \cup_{\text{\textit{p}} \in C} \textbf{Y}_t^{\text{\textit{p}}}$. Within a collaborative unit such as a cell, there will be some tasks that need only be done by a subset of the members of the cell; for example figuring out the logistics or developing certain skills. Thus, the filtered data $\textbf{Z}_t^C$ obtained from the collective data on the cell $\textbf{Y}_t^C$ must be set against these requirements to indicate whether the tasks are being sufficiently completed. Let $\pmb{T}^C$ be the subset of the state space of $X_t^C$ that indicates the set of states considered to be most indicative of an imminent attack by the authorities. One possible measure of collective progress $m_1$ of the cell can then be obtained as
\begin{align}
    m_1 = \sum_{x_i^C \in \pmb{T}^C} \pi_{ti}^C.
\end{align}

\subparagraph{Individual threat:} As discussed above, within a cell, not all tasks need to be performed by each and every member of the cell. Ideally we would like to be able to identify, for each member of a cell $C$, the role that they play within the cell. However, this is not always possible as it requires detailed understanding of the cell's dynamics -- intelligence which is extremely sensitive and difficult to gather \citep{duijn2014relative}. An alternative is to evaluate the threat status of the individuals in $C$ based on their progress on the tasks $\pmb{\vartheta}_t^* \subset \pmb{\vartheta}_t^C$ that \textit{most of the members} of $C$ are expected to have the skills to do. The states for each individual's Lone Terrorist model can be adapted in line with this to obtain the product of measures of individual threat $m_2$ for each member of $C$ as

\begin{align}
    m_2 = \prod_{\text{\textit{p}} \in C} \bigg \{{\sum_{x_i \in \pmb{T}} \pi_{ti}^\text{\textp}} \bigg \}
\end{align}
\noindent where $\pmb{T}$ denotes the set of most dangerous threat states in the individual Lone Terrorist models.

\subparagraph{Latent collaboration:} In any cell, we may not expect each pair to be communicating with each other. However, for any successful collaboration, a certain amount of connectivity is expected between each communicating pair and overall in the cell. Hence we set up two different measures of latent collaboration. For each communicating pair $\{\text{\textit{p}$_i$, \textit{p}$_j$}\}$ in $C$, we measure pairwise cohesion $m_3^*$ as
\begin{align}
    m_3^* = p(\varphi_{ijt} > \ell)
\end{align}
\noindent where $\ell$ is the lower limit of how much we expect each pair to be communicating for the terrorist attack to be enacted. A cell-level measure of pairwise cohesion $m_3$ can be obtained as
\begin{align}
    m_3 = \prod_{\{\text{\textit{p}$_i$, \textit{p}$_j$}\} \in \calP_t \times \calP_t}  p(\varphi_{ijt} > \ell).
\end{align}
Similarly, cell-level cohesion can also be measured through the subnetwork density $m_4$ of $C$ as
\begin{align}
   m_4 = \frac{k}{\binom{n}{2}}
\end{align}
where $k=\abs{E(C_t)}$ represents the number of ties shared by the members of cell $C$ in the network model $\calN_t$ at time $t \geq 0$, $n = \abs{C_t}$ is the size of the cell $C$ and thus $\binom{n}{2}$ is the number of possible ties in $C$.

\subparagraph{Size of the cell:} Although collaborative efforts benefit from sharing resources and skills, a large cell can be unwieldy and increases the risk of the exposure of that cell. For a given type of terrorist attack, the authorities are likely to be able to estimate an ideal cell size $p^*$ either from expert knowledge and intuition or from the literature and reported cases of a similar nature. One simple measure of cell integrity $m_5$ could be obtained as
\begin{align}
   m_5 = \sech(\frac{p - p^*}{p^*}),
\end{align}  
\noindent where $\sech(\cdot)$ is the hyperbolic secant function which only applies a heavy penalisation for larger deviations from $p^*$.\\

We now describe how the above measures $m_i$ for $i = \{1,2, \ldots, 5\}$ may be combined to obtain indicators of a terrorist attack. For a given type of terrorist attack, a cell is most threatening when $m_1 = m_2 = m_3= m_4 = m_5 = 1$. We can obtain an ordered set of indicators of a terrorist attack $\{\varphi_C(i)\}$, $i \in \{0,..,4\}$ as 
\begin{align}
    \varphi_C(i) &= \prod_{j=1}^{5-i} m'_j \\
    \{m'_j\}_{j=1,..,5} &= \sigma(\{m_i\}_{i=1,..5}) \nonumber
\end{align}
where $\sigma$ is a permutation of elements such that for $i=1, \ldots 4$, we have $ 0 \leq m'_{i+1} \leq \ m'_i \leq 1$ and hence for $i=0, \ldots 3$, we have $0 \leq \varphi_C(i) \leq \varphi_C(i+1) \leq 1$. This ordered set is used to check whether the values of one or more measures are overly affecting the base $\varphi_C(0)$ score. Each of these indicators has the property that a higher value of $\varphi_C(i)$ indicates a greater imminence and danger of the threat posed by the cell $C$. Thus, several key factors may be combined to obtain transparent indicators of threat which can guide the counter-terrorism authorities to prioritise and de-prioritise cases. These indicators can be plotted against time to analyse how the threat posed by the cell develops dynamically. Note that the measures used to construct the indicators here can be easily adapted to incorporate other elements that may be considered to be essential by the counter-terrorism authorities, see e.g. \cite{xu2004analyzing, yang2006analyzing}.

\section{Analysis of a Hypothetical Terrorist Group} \label{sec:analysis}

Real-world data regarding the pre-incident activities of frustrated and successful violent plots perpetrated by terrorist groups are often confidential. Hence, here we illustrate the methods from this paper with simulated data for Example \ref{ex:criminal_example}. The data is simulated from some time $t_1$ which is equivalent to time $t+1$ in Example \ref{ex:criminal_example}, and is informed by meetings with relevant policing authorities and publicly available data on various real-world terrorism cases.

\subsection{Lone Terrorist Model for Each Suspect}
\label{Subsec:ind_rve-ex}

Recall that at time $t_1$ the individuals being monitored by the authorities are given by $\calP_{t_1} = \{\text{\textit{p}}_1, \text{\textit{p}}_2, \text{\textit{p}}_3, \text{\textit{p}}_4\}$. Let the sample space of the latent random variable $X_{it}$ in the Lone Terrorist model for suspect \textit{p}$_i$, $1 \leq i \leq 4$ and time $t \geq t_1$ be given by the states $\{\text{`Active'}$, $\text{`Training'}, \text{`Preparing'}, \text{`Mobilised'}, \text{`Neutral'}\}$. Using the criminal profiles of these suspects, and based on their past and current activities, the prior probabilities of the state $X_{i,t_1}$ occupied by these individuals at time $t_1$ are shown in in Figure \ref{fig:ind_rve}. Suspect \textit{p}$_4$ is believed to have received training by pro-terrorist groups and hence has probability weighted toward the $\text{`Training'}$ state, whereas the others have only stated their views and intentions but there is no indication otherwise of them training or preparing, hence they are weighted toward the $\text{`Active'}$ state.

\begin{figure}[ht]
\centering
\begin{subfigure}[t]{0.42\textwidth}
    \includegraphics[scale = 0.35]{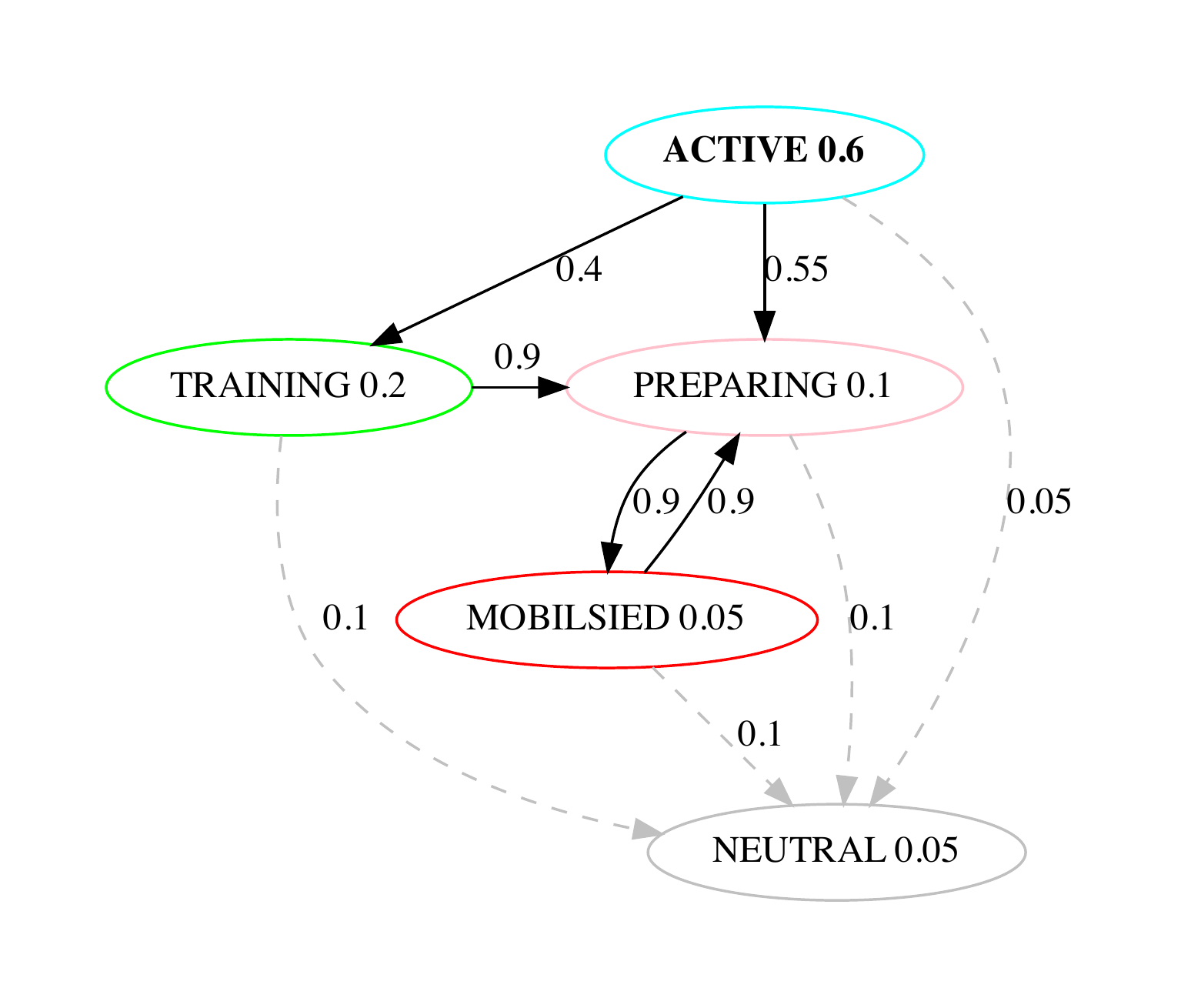}
    \caption{Graph of RVE for \textit{p}$_1$, \textit{p}$_2$ and \textit{p}$_3$.}  
\end{subfigure}
\hfill
\begin{subfigure}[t]{0.42\textwidth}
    \includegraphics[scale = 0.35]{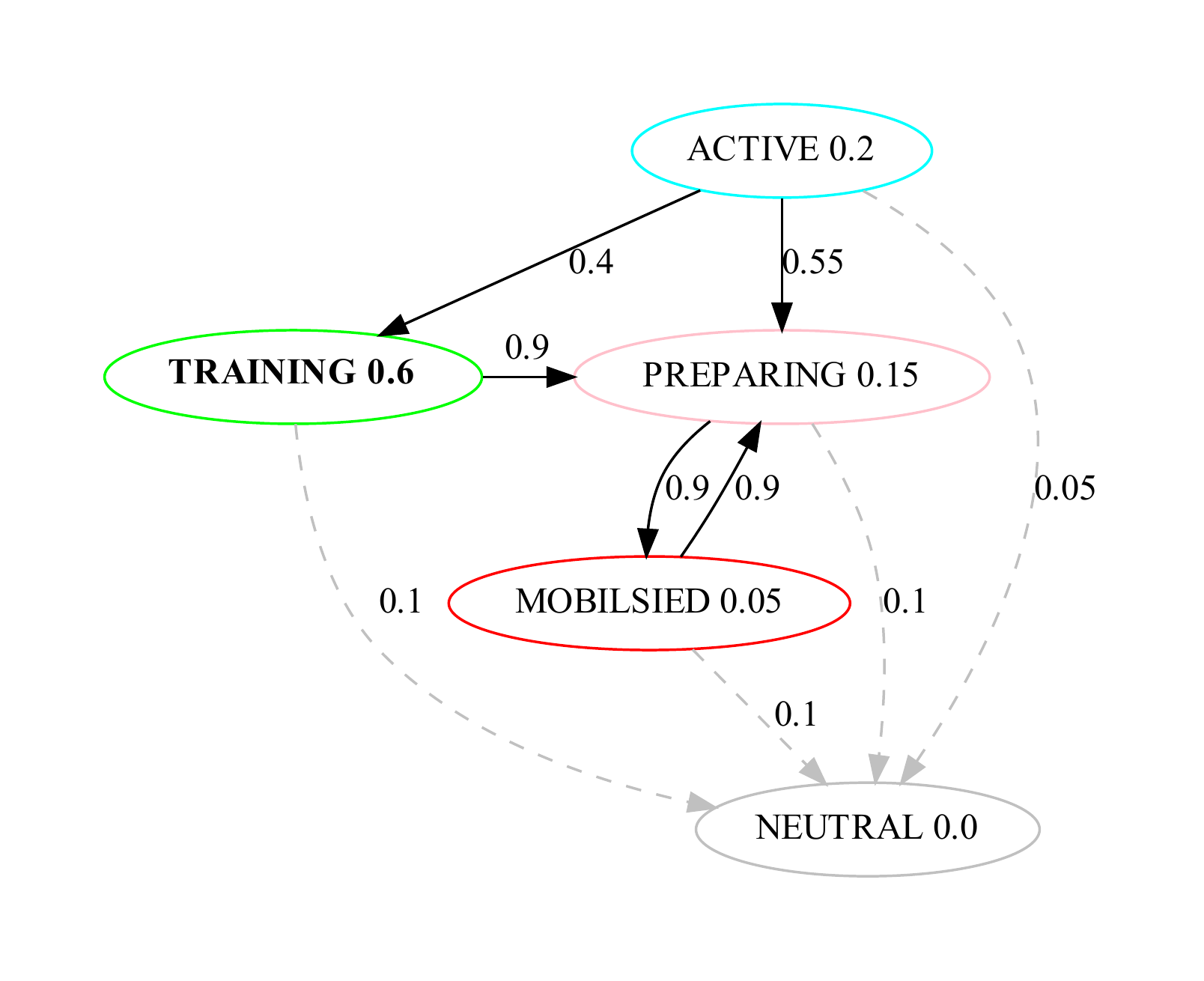}
    \caption{Graph of RVE for \textit{p}$_4$.}
\end{subfigure}
\caption{In both figures, the vertex labels include the prior state probability, edge labels denote the conditional transition probability at time $t_1$.}% and the $\text{`Neutral'}$ state is hidden to prevent visual cluttering.}
\label{fig:ind_rve}
\end{figure}

As these four individuals are in $\calP_{t_1}$, their activities and communications are monitored by the authorities. It is assumed for simplicity that over the ten weeks that follow, the composition of $\calP_{t_1}$ remains unchanged, i.e. none of the existing suspects leave and no new suspects enter this subpopulation. Over the following weeks, it is observed that suspect \textit{p}$_1$'s internet activities include repeated visits to websites of car dealers and car rentals, as well as knife retailers. Their bank account also shows a large influx of funds from an overseas bank account. The internet activity of suspect \textit{p}$_2$ includes visits to illegal bomb making websites, and repeated visits to and comments on extremist radical forums. Suspect \textit{p}$_4$'s internet activity includes searches for online maps and blueprints of government buildings and densely populated commercial areas of the town. Suspect \textit{p}$_4$ is also observed to have physically visited potential bomb testing sites. A full description of the activity data for each of the four suspects observed over a period of ten weeks is given in Appendix C. Using this activity data, the posterior probabilities of $X_{i,t_k}$ are updated in the Lone Terrorist model over the ten weeks for each \textit{p}$_i$ as shown in Figure \ref{fig:rve_ind_evolution} for $i = 1,2,3,4$ and time $1 \leq k \leq 10$. 

%\textcolor{red}{[ADITI: Are the y-axis label and the legend legible?]}
\begin{figure} [h]
\centering
\begin{subfigure}[t]{0.46\textwidth}
    \includegraphics[scale = 0.32, trim = {0cm, 0, 0 ,0}]{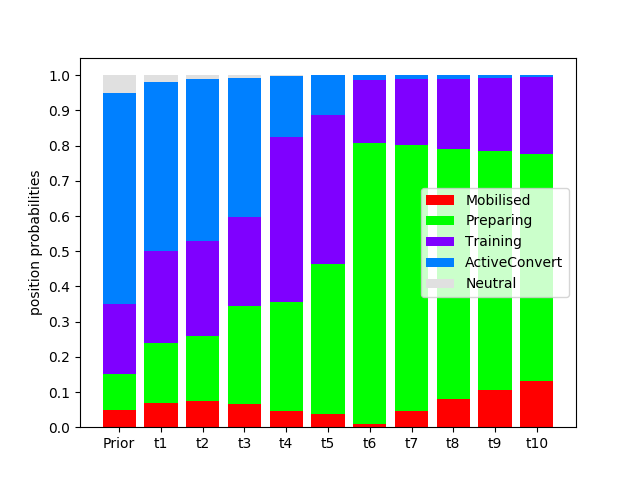}
    \caption{Posterior probabilities for \textit{p}$_1$}
\end{subfigure}
\hfill
\begin{subfigure}[t]{0.46\textwidth}
    \includegraphics[scale = 0.32, trim = {0cm, 0, 0 ,0}]{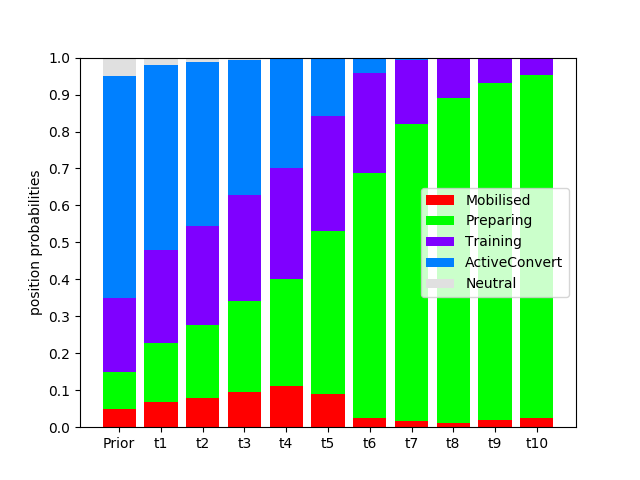}
    \caption{Posterior probabilities for \textit{p}$_2$}
\end{subfigure}
\begin{subfigure}[t]{0.46\textwidth}
    \includegraphics[scale = 0.32, trim = {0cm, 0, 0 ,0}]{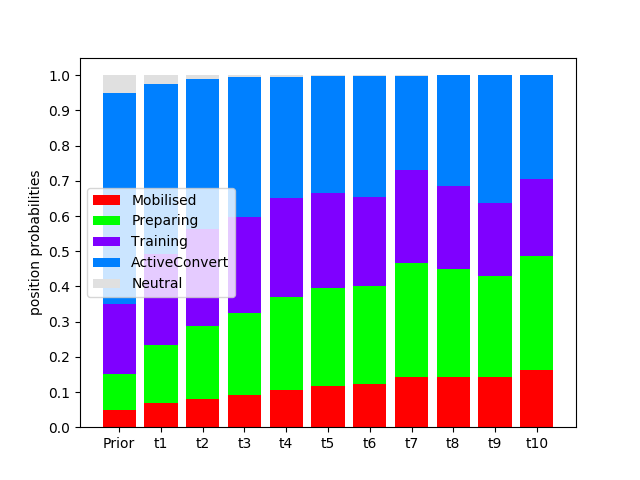}
    \caption{Posterior probabilities for \textit{p}$_3$}
\end{subfigure}
\hfill
\begin{subfigure}[t]{0.46\textwidth}
    \includegraphics[scale = 0.32, trim = {0cm, 0, 0 ,0}]{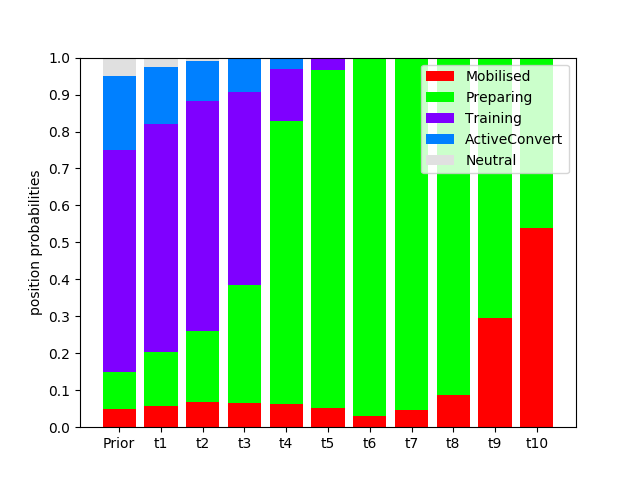}
    \caption{Posterior probabilities for \textit{p}$_4$}
\end{subfigure}
\caption{Posterior threat state probabilities from the Lone Terrorist models of the suspects over the ten weeks.}
\label{fig:rve_ind_evolution}
\end{figure}

\subsection{Terrorist Network Model}
\label{Subsec:network_ex}

Not only are the activities of these suspects being observed, but their communications and interactions with each other are also being recorded over the ten week period. For simplicity, it is assumed here that the pairwise communications data are received from only one information channel: mobile phone calls. The phone call data between a pair of suspects is summarised as the sum of the phone calls in hours between the pair observed over the week. Since we have only one information channel here, we set its efficiency parameter as 1. Table \ref{tab:phone_calls} shows the summary data of the phone calls between each pair in our subpopulation over the ten week period. By week $t_6$, all four individuals share pairwise communications with each other. With this information, we can create edges $e_{1,3}$ and $e_{2,4}$ at time $t_5$, and edge $e_{3,4}$ in the network at time $t_6$ which results in the graph of the network becoming a complete graph. The total time of these phone calls are also observed to increase from weeks $t_7$ to $t_{10}$.

% This data shows that during weeks $t_2$ and $t_3$, the pair $\{\text{\textit{p}}_1, \text{\textit{p}}_2\}$ talk over their mobile phones at a consistent rate. In week $t_3$, \textit{p}$_1$ and \textit{p}$_4$ are observed to converse with each other over a phone call. Further, in week $t_5$, we observe that the pairs $\{\text{\textit{p}}_1, \text{\textit{p}}_3\}$ and $\{\text{\textit{p}}_2, \text{\textit{p}}_4\}$ begin sharing phone conversations. These phone calls lead us to create edges $e_{1,3}$ and $e_{2,4}$ at time $t_5$ in the network to reflect these new ties. 

\begin{table}
\caption{Simulated weekly sum of communication duration data. All the zeros in this table indicate that the pair did not communicate through mobile phone call in that week.}
\label{tab:phone_calls}
\begin{threeparttable}
\begin{tabular}{c|llllll l c|llllll}
\headrow
\thead{Time}&  \thead{$s_{1,2}$} &  \thead{$s_{1,3}$} &  \thead{$s_{1,4}$} &  \thead{$s_{2,3}$} &  \thead{$s_{2,4}$} &  \thead{$s_{3,4}$} & & \thead{Time} &  \thead{$s_{1,2}$} &  \thead{$s_{1,3}$} &  \thead{$s_{1,4}$} &  \thead{$s_{2,3}$} &  \thead{$s_{2,4}$} &  \thead{$s_{3,4}$} \\
\hline
\hiderowcolors
$t_1$  &     0 &     0 &     0 &     0 &     0 &     0 & &  $t_6$ &     5 &     6 &     6 &     5 &     6 &     1 \\
$t_2$  &     3 &     0 &     0 &     1 &     0 &     0 & & $t_7$  &     7 &     6 &     7 &     6 &     7 &     7\\
$t_3$  &     5 &     0 &     2 &     0 &     0 &     0 & & $t_8$ &     6 &     6 &     8 &     4 &     8 &     8\\
$t_4$ &     5 &     0 &     5 &     0 &     0 &     0 & & $t_9$  &     7 &     7 &     9 &     7 &     9 &     9\\
$t_5$  &     5 &     2 &     5 &     0 &     1 &     0 & & $t_{10}$ &     7 &     8 &    11 &     8 &    10 &    10 \\
\hline
\end{tabular}
\end{threeparttable}
\end{table}

Using the phone communication data, the edge weight distributions $\varphi_{i,j,t_k}$ for $i,j = 1,2,3,4$, $i \neq j$ and $1 \leq k \leq 10$ for the Terrorist Network model can be estimated as follows. The prior distributions for $\varphi_{i,j,t_k}$ are set by specifying the $\alpha$ and $\beta$ parameters of the prior Gamma distributions. For instance, based on the prior knowledge the policing authority has on the suspects, they believe that the extent of information shared between \textit{p}$_1$ and \textit{p}$_2$, and between \textit{p}$_2$ and \textit{p}$_3$ is relatively low with some uncertainty at time $t_1$. Hence, the $\alpha$ and $\beta$ parameters are set as 0.7 and 1.41 respectively for $\varphi_{1,2,t_1}$ and $\varphi_{2,3,t_1}$. The setting of the $\alpha$ and $\beta$ parameters for $\varphi_{i,j,t_k}$ for all pairs over the ten week period are given in the supplementary materials. The appendix contains figures showing the evolution of $\varphi_{i,j,t_k}$ through the posterior densities from time $t_3$ to $t_6$. The discount factor $\delta_{i,j,t_k}$ is set to 0.7 across all pairs and for the entire ten week duration so that approximately half the current information is lost over two time steps.

\subsection{Indicators of an Attack}
\label{Subsec:cell_level_threat_ex}

Assuming that the four suspects are working together within a cell, the measures $m_i$ for $i = 2,3,4,5$ described in Section \ref{subsec:indicators} are calculated. Note that for the Terrorist Cell model for measure $m_1$, the task set and observation data are given by the union of the task sets and observation data for the individual Lone Terrorist models of the cell's members. The prior threat state probabilities for the Terrorist Cell model are taken as the corresponding prior probabilities of the suspect within the cell with the highest prior threat, i.e. \textit{p}$_4$. Figure \ref{fig:evolve_rdceg} shows the evolution of the threat state probabilities in the Terrorist Cell model. The posterior probability of the cell being in the $\text{`Preparing'}$ state increases from time $t_5$ as the communications within the cell and the overall activities of the cell increase. Thereafter around time $t_9$, the posterior probability of the cell being in the $\text{`Mobilised'}$ state increases sharply. These measures are combined to obtain the indicators of a terrorist attack $\varphi_C$ as shown in Table \ref{tab:cell_measures_table}. If we were to signal a warning when $\varphi_C(\cdot)$ reaches a certain threshold, say $0.15$, then we can see that for $\varphi_C(0)$ this is not reached till time $t_7$ whereas for $\varphi_C(2)$ this is reached by time $t_3$. In practise, the measures informing these indicators and the chosen thresholds would need to be calibrated using domain experience and judgement. 

\begin{figure} [h]
\centering
\includegraphics[scale = 0.4, trim = {0cm, 0, 0 ,0}]{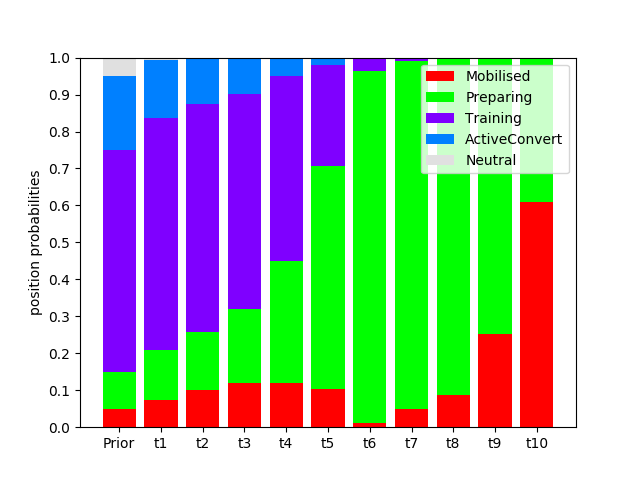}
\caption{Posterior threat state probabilities from the Terrorist Cell model over the ten weeks.}
\label{fig:evolve_rdceg}
\end{figure}

\begin{table}
\caption{Indicators of a terrorist attack for the cell under consideration.}
\label{tab:cell_measures_table}
\begin{threeparttable}
\begin{tabular}{c|rrrrrrrrrrr}
\headrow
{} &  \thead{Prior} &   \thead{$t_1$} &   \thead{$t_2$} &   \thead{$t_3$} &   \thead{$t_4$} &   \thead{$t_5$} &   \thead{$t_6$} &   \thead{$t_7$} &   \thead{$t_8$} &   \thead{$t_9$} &  \thead{$t_{10}$} \\
\hline
m1               &   0.15 & 0.21 & 0.26 & 0.32 & 0.45 & 0.71 & 0.96 & 0.99 & 1.00 & 1.00 & 1.00 \\
m2               &   0.00 & 0.00 & 0.01 & 0.01 & 0.04 & 0.09 & 0.22 & 0.31 & 0.32 & 0.31 & 0.36 \\
m3               &   0.14 & 0.05 & 0.14 & 0.04 & 0.03 & 0.00 & 0.30 & 1.00 & 1.00 & 1.00 & 1.00 \\
m4               &   0.67 & 0.67 & 0.67 & 0.67 & 0.67 & 0.67 & 0.67 & 0.67 & 0.67 & 0.67 & 0.67 \\
m5               &   0.83 & 0.83 & 0.83 & 0.89 & 0.89 & 0.89 & 0.89 & 0.89 & 0.89 & 0.89 & 0.89 \\
\hline
$\varphi_{C}(0)$              &   0.00 & 0.00 & 0.00 & 0.00 & 0.00 & 0.00 & 0.04 & 0.18 & 0.19 & 0.19 & 0.21 \\
$\varphi_C(1)$            &   0.01 & 0.01 & 0.02 & 0.01 & 0.01 & 0.04 & 0.17 & 0.58 & 0.59 & 0.59 & 0.59 \\
$\varphi_C(2)$            &   0.08 & 0.12 & 0.14 & 0.19 & 0.27 & 0.42 & 0.57 & 0.88 & 0.88 & 0.89 & 0.89 \\
$\varphi_C(3)$            &   0.55 & 0.55 & 0.55 & 0.59 & 0.59 & 0.63 & 0.86 & 0.99 & 1.00 & 1.00 & 1.00 \\
$\varphi_C(4)$            &   0.83 & 0.83 & 0.83 & 0.89 & 0.89 & 0.89 & 0.96 & 1.00 & 1.00 & 1.00 & 1.00 \\
\hline
\end{tabular}
\end{threeparttable}
\end{table}

This simple worked example demonstrates how observed activity data and communications data obtained on monitored suspects when combined with prior distributions calibrated to the investigator's knowledge can give real-time indicators of the evolving threat posed by individuals acting in collaboration with others. In the scenario investigated here, driven by the increase in specific activity data and phone call duration, probability of the suspects forming the cell being in either the $\text{`Preparing'}$ or $\text{`Mobilised'}$ states by week $t_{10}$ increased, and correspondingly, the cell as a whole appeared to move from state $\text{`Preparing'}$ occupied since week $t_{5}$ to $\text{`Mobilised'}$ by week $t_{10}$. The indicators of an attack reflected a similar trend, e.g. $\varphi_C(2)$ increased from 0.27 during week $t_{4}$ to 0.57 during week $t_{6}$ and then reaching 0.88/0.89 for weeks $t_{7}$ to $t_{10}$.

\section{Discussion} \label{sec:discussion}

In this paper, we proposed a two-part IDSS combining together the outputs of a Terrorist Network and a collection of hierarchical Lone Terrorist models to aid the counter-terrorism authorities in monitoring the threat posed by individuals acting jointly towards a terrorist attack, and in prioritising and de-prioritising cases. The IDSS can be used to construct indicators of a terrorist attack which facilitate pre-emptive action to frustrate these attacks.

There are several avenues of research that can follow from the work presented in this paper. Recall that the edge weight $\varphi_{ijt}$ along an edge in the Terrorist Network $\calN_t$ is a measure of the pairwise communications shared directly between the suspects \textit{p}$_i$ and \textit{p}$_j$ at time $t$ connected by the edge. This definition of the edge weight then leads to the following conditional independence assumption:
\begin{equation}
        \indep_{\{\text{\textit{p}$_i$, \ \textit{p}$_j$}\} \in \calP_t \times \calP_t} \, \varphi_{ijt} \,\mid\, \mathcal{F}_{t}.
        \label{eq:phi_independence}
    \end{equation}
\noindent Thus, pairwise communications data can be used to estimate the edge weight $\varphi_{ijt}$. For an alternative interpretation of the edge weights as measures of the extent of collaboration between the individuals connected by the edge, the above conditional independence statement does not hold. For instance, the extent of collaboration $\varphi_{ijt}$ between \textit{p}$_i$ and \textit{p}$_j$ is also affected by the communications and interactions they both share with a common neighbour, say \textit{p}$_k$. This would additionally lead to the violation of the output independence assumption described in Statement \ref{eq:output_independence}. Under this independence structure, we would no longer be able to estimate $\varphi_{ijt}$ for each pair \textit{p}$_i$ and \textit{p}$_j$ independently. In this case, the decouple/recouple strategy introduced by \citet{gruber2016gpu, zhao2016dynamic} for financial and economic multivariate time-series applications could be explored, although the recoupling is unlikely to be closed form and would need to be numerically estimated. Here, any gains achieved by refining the interpretation of $\varphi_{ijt}$ to include collaboration in a broader sense would need to be weighed against the loss in transparency and interpretability due to the numerical estimation.

Another avenue of research would be to incorporate link detection within the Terrorist Network using existing link detection methods (see Section \ref{subsec:related_work}) to identify potentially hidden ties. This work can further be extended by developing a bespoke clustering algorithm using domain-specific stochastic set functions (as have been used in the criminology literature, see e.g. \citet{wang2013learning}) to identify previously unknown cells or monitor the evolution of new cells within the network.

Finally, the generic architecture of an IDSS using the decoupling methodology might be applicable to other domains where there is a requirement to integrate individual time-series with dynamic interactions among individuals, modelled by a network, who collaborate to realise a shared objective. Examples of this include social processes within politics, governments or communities where complex interacting individuals have shared objectives.

\bibliography{sample}

% \begin{biography}[example-image-1x1]{A.~One}
% Please check with the journal's author guidelines whether author biographies are required. They are usually only included for review-type articles, and typically require photos and brief biographies (up to 75 words) for each author.
% \bigskip
% \bigskip
% \end{biography}

% \graphicalabstract{example-image-1x1}{Please check the journal's author guildines for whether a graphical abstract, key points, new findings, or other items are required for display in the Table of Contents.}

\end{document}

\maketitle

\renewcommand{\thesection}{\Alph{section}}
\section{State Space and Dynamics of the Integrated Model}
\subsection{Individual or Group Attack Model} \label{Sec:RVE}
We work within the standard probability model of a filtered probability space
\begin{align*}
\{\Omega,\calF,\mathbb{F},\mathbb{P}\}
\end{align*}
where $\Omega$ is the sample space, $\calF$ is a $\sigma$-field on $\Omega$, $\mathbb{F}$ is a filtration of non-decreasing subsets of $\calF$, and $\mathbb{P}$ is a probability measure on $\{\Omega,\calF\}$.

Let $X_t$ be the threat state of a individual or a collaborating group of individuals, represented as a latent discrete state continuous time stochastic process; 
$\pmb{\vartheta}_t$ the latent vector of binary valued tasks that progresses $X_t$ towards an attack; $\textbf{Y}_t$ the observed data produced by the task activities; and $\textbf{Z}_t$ a deterministic function of $\textbf{Y}_t$ that filters the data into a signal that informs the value of $\pmb{\vartheta}_t$. Latent here means that the value of a variable at time $t$ is not usually observed even at time $t$, e.g., $\sigma(X_t) \not\subset \calF_t$ where $\sigma(X_t)$ is the $\sigma$-field generated by $X_t$. With minor abuses of notation to aid understanding\footnote{Strictly $\textbf{Y}: \Omega \times \{0,1\}^R \times \R_+ \rightarrow \R^n$ and $\textbf{Z}: \R^n \times \R_+ \rightarrow \R^R$.}
\begin{align*}
X: &\Omega \times \R_+ \rightarrow \calX = \{x_0, \ldots, x_{m-1}\} \\
\pmb{\vartheta}: &\Omega \times \calX \times \R_+ \rightarrow \{0,1\}^R\\
\textbf{Y}: &\Omega \times \pmb{\vartheta} \times \R_+ \rightarrow \R^n \\
\textbf{Z}: &\textbf{Y} \times \R_+ \rightarrow \R^R
\end{align*}

Let $\textbf{I}_i$ be an index set of the indices of the tasks relevant for the the $i$th state $x_i$ such that the event that $X_t=x_i$ is independent of the tasks not in $\pmb{\vartheta}_{I_i}$ given the tasks in $\pmb{\vartheta}_{I_i}$:
\begin{align*}
\textbf{I}_i &\subset \{1,\ldots,R\} \\
\pmb{\vartheta}_{I_i} &= \{\vartheta_j\}_{j\in I_i} \\
\{X_t=x_i\} \; &\indep\; \{\vartheta_j\}_{j \not\in I_i} \ | \ \{\vartheta_j\}_{j \in I_i}
\end{align*}

Having elicited assumptions of the dependency structure between variables from domain experts, we construct that variables $X_t$, $\pmb{\vartheta}_t$, $\textbf{Y}_t$, and the signals $\textbf{Z}_t$ such that the below conditional independence statements hold
\begin{align*}
\pmb{\vartheta}_t \; &\indep \; \textbf{Y}_t \ | \ \textbf{Z}_t \\ 
X_t \; &\indep\; \{\textbf{Y}_t, \textbf{Z}_t\} \ | \ \pmb{\vartheta}_t
\end{align*}

The generative dynamics of $X_t$ is determined by the semi-Markov transition matrix $\calM$ in Equation \ref{Eqtn:semiMarkov} as
\begin{align}
\label{Eqtn:semiMarkov}
\calM(s,t) = \calM_0 \circ \calM_1 \in \R_+^{m \times m} \\
\calM_0 = [ p_{i,j} ]_{\{0 < i, j \leq m - 1\}} \nonumber\\
\calM_1 = [ q_{i,j}(s,t) ]_{\{0 \leq i, j \leq m - 1; 0 \leq u < v < \infty\}} \nonumber
\end{align}
where $p_{i,j}$ is the probability of transition from $x_i$ to $x_j$; so that $\calM_0$ is the Markov transition matrix for the embedded Markov chain of the semi-Markov process (see \citet{Cinlar1975}); and $q_{i,j}(s,t)$ is the holding time distribution at time $s$ for the random time $t-s$ that $X_t$ stays in state $x_i$ before transitioning to state $x_j$ given such a transition will occur; and $\circ$ is the element-wise product operator.

The generative process of $X_t$ is thus given by Equation \ref{Eqtn:SemiMarkovDynamics} as follows
\begin{align}
\label{Eqtn:SemiMarkovDynamics}
\pmb{\pi}(X_t|\calF_s) =\calM(s,t) \; \pmb{\pi}(X_s|\calF_s)
\end{align}
The updating of the state probability vector $\pmb{\pi}$ having calculated the signal $\textbf{z}_t$ from the observed data $\textbf{y}_t)$ is given by Equation \ref{Eqtn:SemiMarkovUpdate} as follows
\begin{align}
\label{Eqtn:SemiMarkovUpdate}
\pmb{\pi}(X_t|\calF_t) &= \pmb{\pi}(X_t|\textbf{z}_t) \nonumber \\ 
&= \pmb{\pi}(X_t|\calF_s) \circ \frac{\mathbb{P}(\textbf{z}_t|X_t)}{\mathbb{P}(\textbf{z}_t)} \nonumber \\
&\propto \pmb{\pi}(X_t|\calF_s) \circ \mathbb{P}(\textbf{z}_t|X_t) \nonumber \\
&= \{\; \calM(s,t) \; \pmb{\pi}(X_s|\calF_s) \; \} \circ \mathbb{P}(\textbf{z}_t|X_t)
\end{align}
where again $\circ$ is the element-wise product operator, here for two vectors, and
\begin{align*}
&\mathbb{P}(\textbf{z}_t|X_t) = \{\mathbb{P}(\textbf{z}_t|X_t=x_i)\}_{i=0,\ldots,m-1} \in \mathbb{R}^m \\
&\mathbb{P}(\textbf{z}_t|X_t=x_i) = \sum_{\pmb{\vartheta}_{I_i} \in \{0,1\}^{\abs{I_i}}}  \mathbb{P}(\textbf{z}_t|\pmb{\vartheta}_{I_i}) \mathbb{P}(\pmb{\vartheta}_t|X_t=x_i) \in \mathbb{R}\\
%&R_i = \abs{I_i} \\
&\mathbb{P}(\textbf{z}_t|\pmb{\vartheta}_t) \in \mathbb{R} \\
&\mathbb{P}(\pmb{\vartheta}_t|X_t) \in \mathbb{R}^m \\
&\mathbb{P}(\pmb{\vartheta}_t|X_t=x_i) \in \mathbb{R}
\end{align*}

The state space $\calX$ of $X_t$ forms the vertices of the Reduced Dynamic Chain Event Graph $\mathfrak{G}$ (see \citet{shenvi2019}). The edges of $\mathfrak{G}$ are determined by the non-zero elements of the semi-Markov transition matrix $\calM$.

\subsection{Attack Group Network}
Let $\calP_t^*$ be the population of persons of interest at time $t$ and $\calN_t$ be an undirected weighted network over a subset of $\calP_t^*$. 
\begin{align*}
\calN_t &= (V(\calN_t),E(\calN_t))  \\
V(\calN_t) &= \calP_t \subset \calP_t^*  \\
E(\calN_t) &= \{e_{i,j}\} \subset \{ (\textit{p}_i,\textit{p}_j): \textit{p}_i, \textit{p}_j \in V(\calN_t)\}
\end{align*}
The vertices of $\calN$ are some set of individuals $\{\textit{p}_i\}_{i=1,\ldots, n} \subset \calP_t^*$. Each of these $\textit{p}_i$ has a dynamic threat state, namely $X_t(\textit{p}_i)$ and collectively they have a group threat level $X_t(C)$. Both $X_t(\textit{p}_i)$ and $X_t(C)$ have the dynamics as defined in section \ref{Sec:RVE}. The weight of the edge $e_{i,j}$ is denoted $\varphi_{i,j}$ and is itself a latent stochastic process
\begin{align*}
\varphi: \Omega \times E(\calN_t) \rightarrow \R_+
\end{align*}
Extending the domain of $\varphi$ to include all $(\textit{p}_i,\textit{p}_j) \in V(\calN_t) \times V(\calN_t)$ by setting $\varphi_{i,j} = 0$ when $(\textit{p}_i,\textit{p}_j) \not\in E(\calN_t)$ we obtain a positive symmetric square matrix.
\begin{align*}
\Phi &= [ \varphi_{i,j} ] \in \R^{N \times N} \\
N &= \abs{V(\calN_t)}
\end{align*}
%\textcolor{blue}{Add X-ref Section 4}
Raw communications data, $\pmb{R}_{i,j}$ (see Section 4 of the main article) are observed on $K$ channels for each edge $e(i,j)$ of $\calN$ and are scaled and rounded to integer values $\pmb{S}_{i,j}$. This raw data, at any time $t$, is a non-deterministic function of the actual latent level of pairwise pairwise communications $\varphi_{i,j,t}$ between $\textit{p}_i$ and $\textit{p}_j$. We indicate the random nature of $\pmb{R}_{i,j}$ by showing that it is also a function of the sample space $\Omega$ and, for any particular value, a function of the \textit{state of the world} abstractly represented by the sample point $\omega$ in the sample space. In contrast $S$ is a deterministic function of $R$.
\begin{align*}
\pmb{R}_{i,j}: \Omega \times \R \times \R_+ \rightarrow \R^k \\
\pmb{R}_{i,j}(\omega,\varphi_{i,j},t) \in \R^K \\
\pmb{S}_{i,j}: \R^K \rightarrow \mathbb{Z}^K  \\
\pmb{S}_{i,j}(\pmb{R}_{i,j}) \in \mathbb{Z}^K
\end{align*}
Thus we obtain a integer vector valued matrix of the same dimension as $\Phi$ where the $(i,j)$th element is the zero $k$-vector if $(\textit{p}_i,\textit{p}_j) \not\in E(\calN)$.
\begin{align*}
S &= [ \textbf{s}_{i,j} ] \in \mathbb{Z}^{K \times N \times N}
\end{align*}

For a specified sequence of times $t_0, t_1, \ldots $ we model the latent variables $\varphi_{i,j,t_i}$ with Gamma distributions and the components of the vector $\pmb{s}_{i,j}$ with Poisson distributions conditional $\varphi_{i,j,t_i}$ as below, where $\xi_k$ is the efficiency of channel $k \in \{1,\ldots, K\}$. Since the Gamma distribution is the conjugate prior for the Poisson distribution the recurrence maintains closed form through periodic updates; and the $\alpha$ and $\beta$ parameters give the intuitive interpretation of $\alpha$ units of communication over $\beta$ time units.
\begin{align*}
\varphi_{i,j}(t_0) &\sim Gamma(\alpha_{i,j,t_0},\beta_{i,j,t_0}) \\
S_{i,j,k,t_i} &\sim Poisson( \lambda_{i,j,k}(t_i))  \\
\lambda_{i,j,k}(t_i) &= \xi_k \varphi_{i,j,t_i} \\
\alpha_{i,j,t_i} &= \delta_{i,j} \alpha_{i,j,t_{i-1}} + \sum_k s_{i,j,k,t_i} \\
\beta_{i,j,t_i} &= \delta_{i,j} \beta_{i,j,t_{i-1}} + \sum_k \xi_{k} \\
\varphi_{i,j,t_i} &\sim Gamma(\alpha_{i,j,t_i},\beta_{i,j,t_i}) \\
\end{align*}
where the k-vector $\pmb{s}_{i,j,t_i}$ is calculated from the observed k-vector of raw data $\pmb{r}_{i,j,t_i}$ and the $\delta_{i,j}$ are fixed discount factors.
\subsection{Measures of Group Threat}
The threat information contained in the network $\calN$ consists of 
\begin{itemize}
\item The individual threat levels $\pmb{\pi}(X_t(\textit{p}_i))$, for $\textit{p}_i \in V(\calN)$, that are based on the data series $\textbf{Y}_t(\textit{p}_i)$; 
\item The group threat level $\pmb{\pi}^C(X_t(C))$, that is based on the data series $\textbf{Y}^C_t$ emanating from the group of individuals $\textit{p}_i \in V(\calN)$ (see Section 5.1 of the main article); and
%\textcolor{blue}{Xref Section 5.1}
\item The pairwise communication levels $\varphi_{i,j}$ along edges $e_{i,j} \in E(\calN)$ that comprise the latent communication matrix $\Phi$ for the network $\calN$. Note that $\Phi$ itself is not observable; we can, however, observe the inferred distribution of the elements of $\Phi$ from the assumption that each element follows a $Gamma(\alpha_{i,j},\beta_{i,j})$ distribution and the values of $\alpha_{i,j}$ and $\beta_{i,j}$ at any time $t$. 
\end{itemize}

Real valued functions of this set of information provide summary measures of the overall threat level of the group and can be used to prioritise groups and de-prioritise groups for close monitoring amongst a set of groups within the whole population. Let $m$ be a real valued function over this set of information such that:
\begin{align}
m: \R^{m\times N} \times \R^m \times \R^{2 \times N \times N} &\rightarrow \R \\
m(\pmb{\pi}(X_t(\textit{p}_i)),\pmb{\pi}^{C}(X_t(C)),[ (\alpha_{i,j,t},\beta_{i,j,t}) ]_{e_{i,j}\in E(\calN)} \in \R
\end{align}
%\textcolor{blue}{Xref Section5.1 and 6.3}
Then for each distinct communicating network, $\calN_1, \ldots, \calN_n \subset \calP_t^* $, within the population of persons of interest, the scores $m_1, \ldots, m_n$ can be used to select the $0 \leq k \leq n$ groups to monitor that optimises the authorities utility function given their constraints, see Sections 5.1 and 6.3 of the main article.

\section{Independent updating of the edge weight variables}

%add cross-ref from main text if this is combined with the main document
With the standard first-order Markovianity and output independence assumptions given in Statements 5 and 6 in the main article, the relationship between the matrices $S_t$ and $\Phi_t$ can then be represented by a 2-time-slice dynamic Bayesian network (DBN) \citep{dean1989model} whose graph is shown in Figure \ref{fig:2ts-dbn}. This enables us to estimate the edge weight random variable $\varphi_{ijt}$ using observational data $\textbf{s}_{ijt}$ for each pair \ti \ and \tj \ independently as proved below.
%\textcolor{blue}{Need to amend pdf's to use s/S rather than u/U for consistency with my changes throughout}
\begin{figure}[h]
    \centering
    \begin{subfigure}[b]{0.45\textwidth}
    \centering
        \includegraphics[trim = 3cm 0cm 3cm 0cm, scale = 0.25 ]{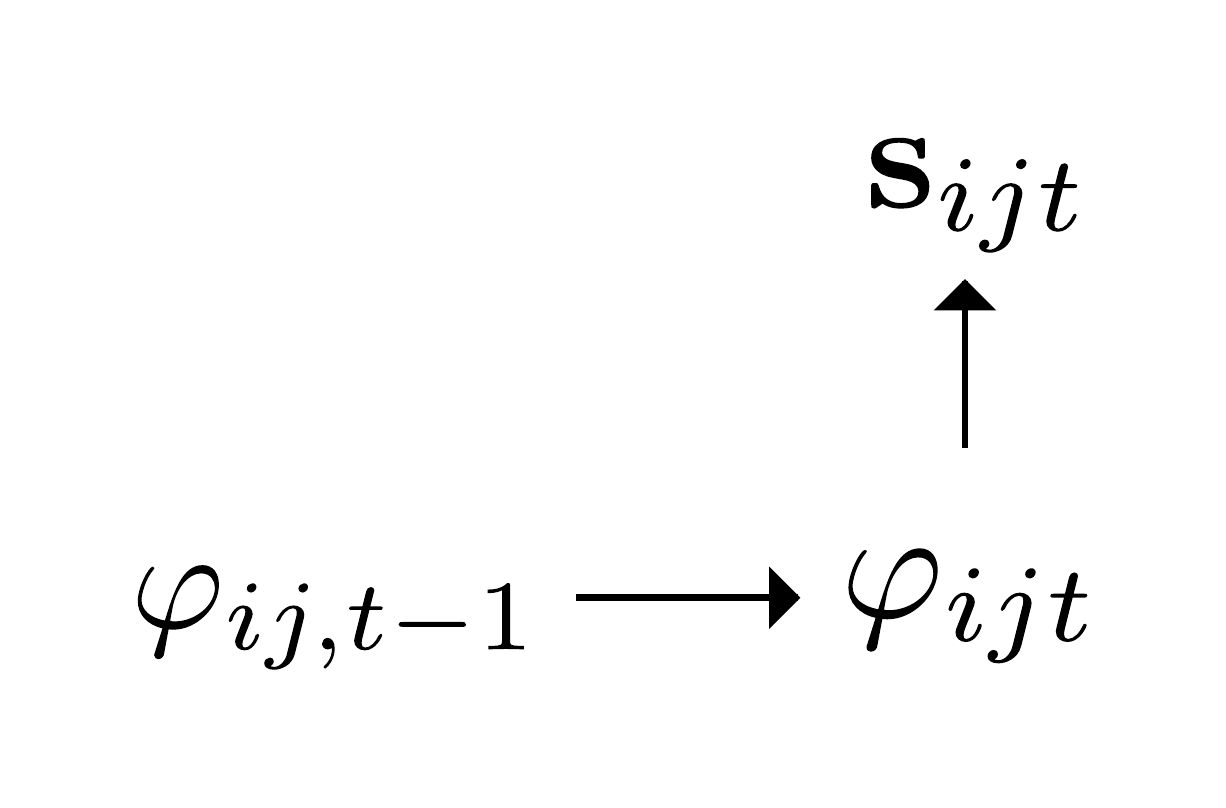}
        \caption{}
    \end{subfigure}
    \begin{subfigure}[b]{0.45\textwidth}
    \centering
        \includegraphics[trim = 3cm 0cm 3cm 0cm, scale = 0.25 ]{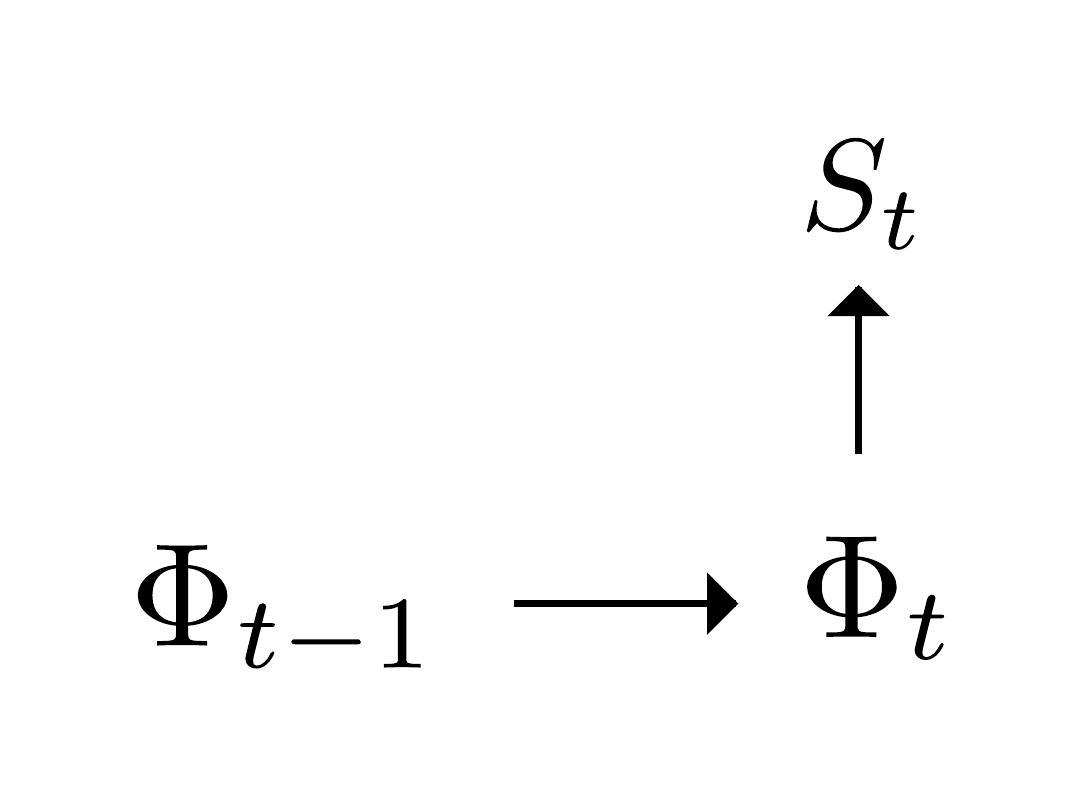}
        \caption{}
    \end{subfigure}
    \caption{Graphical representation of the 2-time-slice DBN: (a) on a univariate level; (b) on a multivariate level (labelled as graph $\mathcal{G}$).}
        \label{fig:2ts-dbn}
\end{figure}

\begin{theorem} 
\label{thm:likelihood_updates}
The marginal likelihood of the 2-time-slice DBN represented by the graph $\mathcal{G}$ decomposes into the product of the one-step-ahead forecasts. Additionally all elements of ${S}_t$ and ${\Phi}_t$ can be updated independently.
\end{theorem}

\begin{proof}
At time $t \geq 0$, recall that $\mathcal{F}_t$ denotes all past data $S_{t'}$ and edge weight random variables $\Phi_{t'}$ for $t' < t$. Denote by $p(\varphi_{ijt} \,|\, \calF_{t})$ the prior distribution for $\varphi_{ijt}$ given the past information till time $t-1$. Since the edge weight random variable $\varphi_{\cdot, \cdot, t}$ is a quantitative measure of the information being shared between a pair of suspects at time $t$, $\varphi_{ijt}$ for each pair $\{\text{\ti, \tj}\} \in \Omega_t \times \Omega_t$ at time $t$ can be estimated independently. Thus, the prior density of $\Phi_t$ can be written as
\begin{equation}
    p(\Phi_t \,|\, \calF_{t}) = \prod_{e_{i,j} \in E(\calN)} p(\varphi_{ijt}\,|\, \calF_{t}).
    \label{eq:priors}
\end{equation}

With the first-order Markovianity (Statement 5 in the main article) and the output independence assumptions (Statement 6 in the main article), the matrix $S_t$ and the posterior of $\Phi_t$ decompose as follow
\begin{equation}
    p(S_t \,|\, \calF_{t}) = \prod_{e_{i,j} \in E(\calN)} p(\textbf{s}_{ijt} \,|\, \calF_{t}),
\end{equation}
\begin{equation}
    p(\Phi_t \,|\,S_t, \calF_{t}) = \prod_{e_{i,j} \in E(\calN)}  p(\varphi_{ijt} \,|\, \textbf{s}_{ijt}, \calF_{t}).
\end{equation}

For each pair \{\ti, \ \tj\}, on observing a vector $\textbf{s}_{ijt}$ of data from the $K$ communication channels, $\varphi_{ijt}$ can be updated as follows due to Statement 4 in the main article,
\begin{align}
    p(\varphi_{ijt} \,|\, \textbf{s}_{ijt}, \calF_{t}) &\propto p(\varphi_{ijt} \,|\, \calF_{t}) p(\textbf{s}_{ijt} \,|\, \varphi_{ijt}, \calF_{t}) \notag \\
    & = \prod_{k=1}^K  p(\varphi_{ijt} \,|\, \calF_{t})  p(\textbf{s}_{ijkt}\,|\, \varphi_{ijt}, \calF_{t}).
    \label{eq:posterior}
\end{align}

Thus the one-step-ahead forecasts can be written as
\begin{align}
    p(\textbf{s}_{ijt} \,|\, \calF_{t}) &= \int_{\varphi_{ijt}} p(\textbf{s}_{ijt} \,|\, \varphi_{ijt}, \calF_{t}) p(\varphi_{ijt} \,|\, \calF_{t}) \, d\varphi_{ijt} \notag \\
    &= \prod_{k=1}^K \int_{\varphi_{ijt}} p(s_{ijkt} \,|\, \varphi_{ijt}, \calF_{t}) p(\varphi_{ijt} \,|\, \calF_{t}) \, d\varphi_{ijt}
\end{align}

Now the marginal likelihood of the 2-time-slice DBN model described by the graph $\mathcal{G}$ can be decomposed into a product of the one-step-ahead forecasts as follows:
\begin{align}
p(S_{1}, \ldots, S_{t} \,|\, \calF_1) &= \prod_{s=1}^t p(S_{s} \,|\, \calF_{s}) \notag \\
&= \prod_{s=1}^t \prod_{e_{i,j} \in E(\calN)} \prod_{k=1}^K \int_{\varphi_{ijs}} p(s_{ijks} \,|\, \varphi_{ijs}, \calF_{s}) p(\varphi_{ijs} \,|\, \calF_{s}) \, d\varphi_{ijs}
\end{align}
or equivalently as a sum of the log marginal likelihoods as
\begin{align}
\log (p(S_{1}, \ldots, S_{t} \,|\, \calF_1)) = \sum_{s=1}^t \sum_{e_{i,j} \in E(\calN)} \sum_{k=1}^K  \log p(s_{ijks} \,|\, \calF_{s})
\end{align}
\noindent where $\calF_1$ reflects the prior information used to set up the model at time $t_0$. 

\end{proof}

\section{Additional information for the illustrative example}

Figure \ref{fig:observed_data} shows the activity data for each of the four suspects observed over a period of ten weeks. This activity data is used in the individual Lone Terrorist models for the suspects as described in Section 3.1 of the main article and Appendix \ref{Sec:RVE}. 

\begin{figure} [h]
\centering
\begin{subfigure}[t]{0.48\textwidth}
    \includegraphics[width=\textwidth]{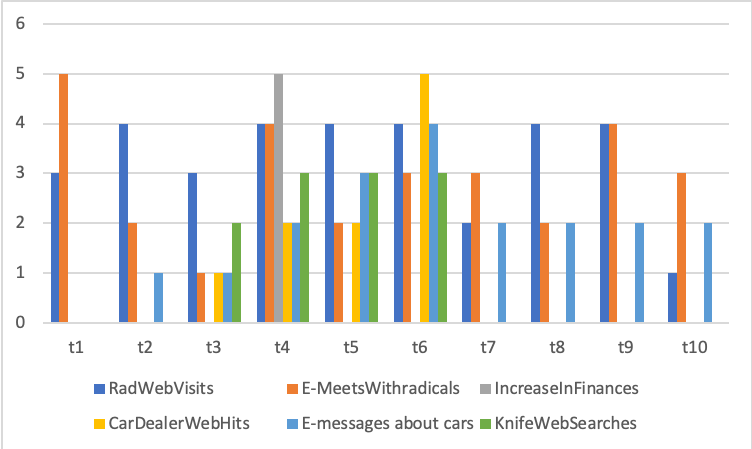}
    \caption{Activity data for \textit{p}$_1$}
\end{subfigure}
\hfill
\begin{subfigure}[t]{0.48\textwidth}
    \includegraphics[width=\textwidth]{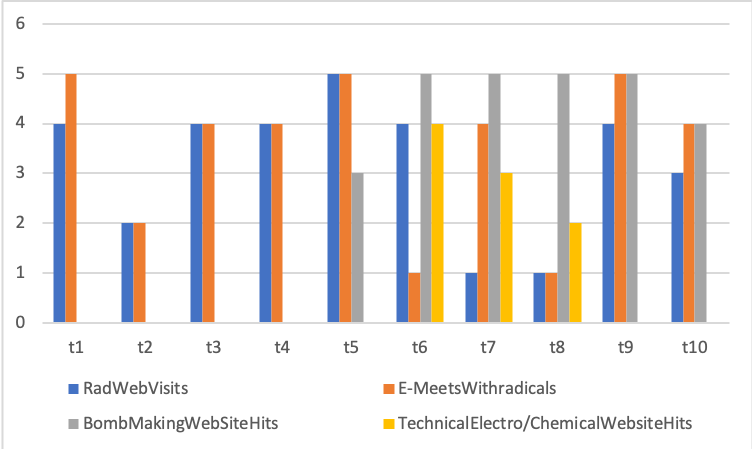}
    \caption{Activity data for \textit{p}$_2$}
\end{subfigure}
\begin{subfigure}[t]{0.48\textwidth}
    \includegraphics[width=\textwidth]{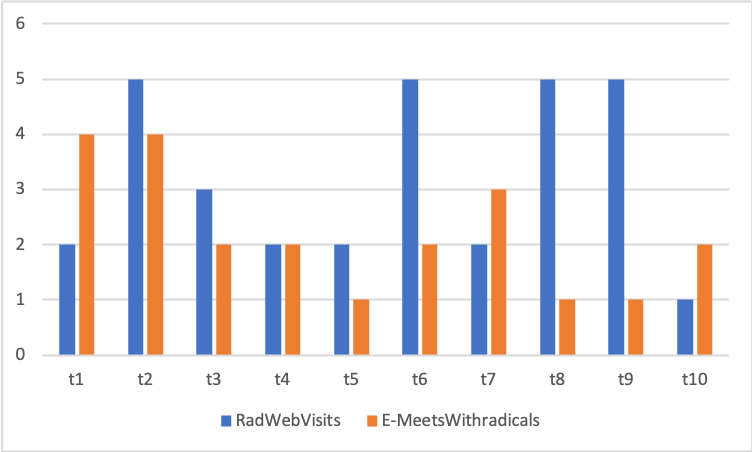}
    \caption{Activity data for \textit{p}$_3$}
\end{subfigure}
\hfill
\begin{subfigure}[t]{0.48\textwidth}
    \includegraphics[width=\textwidth]{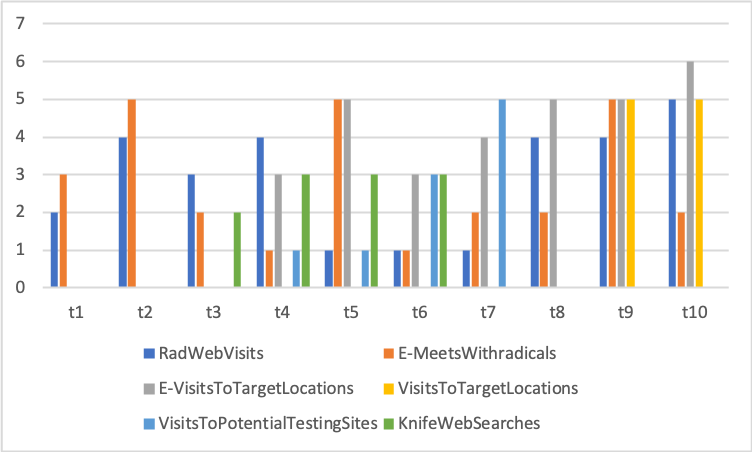}
    \caption{Activity data for \textit{p}$_4$}
\end{subfigure}
\caption{Activity data for the four suspects over the observed time period of ten weeks.}
\label{fig:observed_data}
\end{figure}

The setting of the $\alpha$ and $\beta$ parameters for $\varphi_{i,j,t_k}$ for all pairs over the ten week period are shown in Table \ref{tab:parameters}, and the evolution of $\varphi_{i,j,t_k}$ is shown through the posterior densities in Figure \ref{fig:posterior_densities}. The discount factor $\delta_{i,j,t_k}$is set to 0.7 across all pairs and for the entire ten week duration.

\begin{table}[H]
\centering
\begin{footnotesize}
\begin{tabular}{lllllllllllll}
\hline
{} &   $\varphi_{1,2}$ &  &   $\varphi_{1,3}$ &  &   $\varphi_{1,4}$ &  &   $\varphi_{2,3}$ &  &   $\varphi_{2,4}$ &  &   $\varphi_{3,4}$ &  \\
\hline
&  $\alpha$ &       $\beta$ &  $\alpha$ &       $\beta$ &  $\alpha$ &       $\beta$ &  $\alpha$ &       $\beta$ &  $\alpha$ &       $\beta$ &  $\alpha$ &        $\beta$ \\
\hline
$t_1$ prior  &   0.70 &       1.41 &     &         &     &         &   0.70 &       1.41 &     &         &     &          \\
$t_1$ post   &   0.70 &       2.41 &     &         &     &         &   0.70 &       2.41 &     &         &     &          \\
$t_2$ prior  &   0.50 &       1.70 &     &         &     &         &   0.50 &       1.70 &     &         &     &          \\
$t_2$ post   &   3.50 &       2.70 &     &         &     &         &   1.50 &       2.70 &     &         &     &          \\
$t_3$ prior  &   2.46 &       1.90 &     &         &     &         &   1.05 &       1.90 &     &         &     &          \\
$t_3$ post   &   7.46 &       2.90 &     &         &      2 &          1 &   1.05 &       2.90 &     &         &     &          \\
$t_4$ prior  &   5.26 &       2.04 &     &         &   1.41 &       0.70 &   0.74 &       2.04 &     &         &     &          \\
$t_4$ post   &  10.26 &       3.04 &     &         &   6.41 &       1.70 &   0.74 &       3.04 &     &         &     &          \\
$t_5$ prior  &   7.23 &       2.15 &     &         &   4.52 &       1.20 &   0.52 &       2.15 &     &         &     &          \\
$t_5$ post   &  12.23 &       3.15 &      2 &          1 &   9.52 &       2.20 &   0.52 &       3.15 &      1 &          1 &     &          \\
$t_6$ prior  &   8.62 &       2.22 &   1.41 &       0.70 &   6.71 &       1.55 &   0.37 &       2.22 &   0.70 &       0.70 &     &          \\
$t_6$ post   &  13.62 &       3.22 &   7.41 &       1.70 &  12.71 &       2.55 &   5.37 &       3.22 &   6.70 &       1.70 &      1 &           1 \\
$t_7$ prior  &   9.60 &       2.27 &   5.22 &       1.20 &   8.95 &       1.80 &   3.78 &       2.27 &   4.72 &       1.20 &   0.70 &        0.70 \\
$t_7$ post   &  16.60 &       3.27 &  11.22 &       2.20 &  15.95 &       2.80 &   9.78 &       3.27 &  11.72 &       2.20 &   7.70 &        1.70 \\
$t_8$ prior  &  11.70 &       2.30 &   7.91 &       1.55 &  11.24 &       1.97 &   6.89 &       2.30 &   8.26 &       1.55 &   5.43 &        1.20 \\
$t_8$ post   &  17.70 &       3.30 &  13.91 &       2.55 &  19.24 &       2.97 &  10.89 &       3.30 &  16.26 &       2.55 &  13.43 &        2.20 \\
$t_9$ prior  &  12.47 &       2.33 &   9.80 &       1.80 &  13.56 &       2.09 &   7.68 &       2.33 &  11.46 &       1.80 &   9.46 &        1.55 \\
$t_9$ post   &  19.47 &       3.33 &  16.80 &       2.80 &  22.56 &       3.09 &  14.68 &       3.33 &  20.46 &       2.80 &  18.46 &        2.55 \\
$t_{10}$ prior &  13.72 &       2.34 &  11.84 &       1.97 &  15.90 &       2.18 &  10.34 &       2.34 &  14.42 &       1.97 &  13.01 &        1.80 \\
$t_{10}$ post  &  20.72 &       3.34 &  19.84 &       2.97 &  26.90 &       3.18 &  18.34 &       3.34 &  24.42 &       2.97 &  23.01 &        2.80 \\
\hline
\end{tabular}
\caption{Evolution of the prior and posterior parameters for $\varphi_{i,j,t_k}$ during the ten weeks from $t_1$ to $t_{10}$.}
\label{tab:parameters}
\end{footnotesize}
\end{table}

\begin{figure}[H]
\centering
\begin{subfigure}[t]{0.42\textwidth}
    \includegraphics[width=\textwidth]{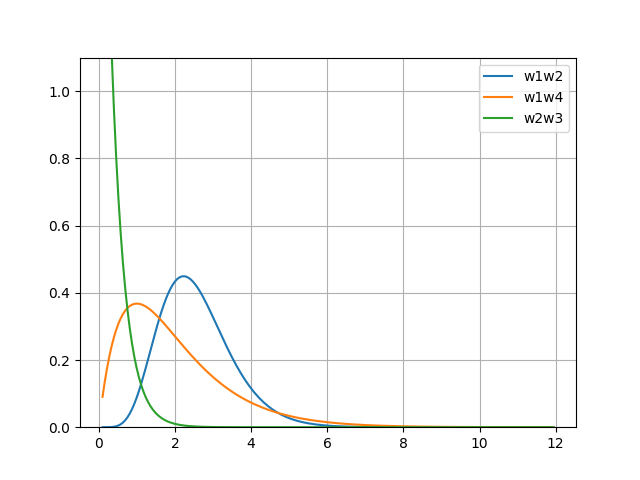}
    \caption{Posterior density at $t_3$}
\end{subfigure}
\begin{subfigure}[t]{0.42\textwidth}
    \includegraphics[width=\textwidth]{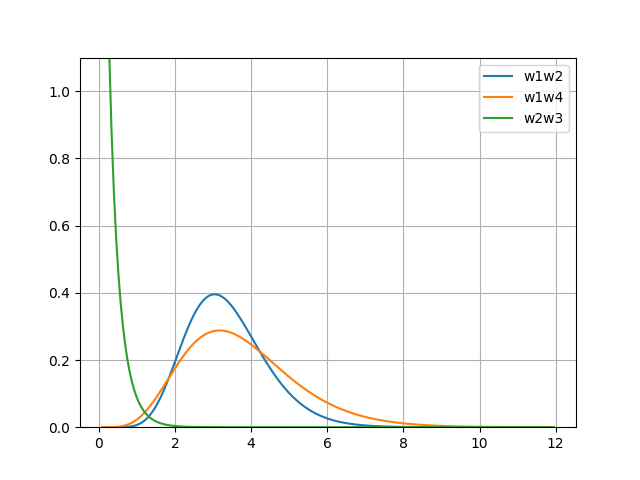}
    \caption{Posterior density at $t_4$}
\end{subfigure}
\vfill
\begin{subfigure}[t]{0.42\textwidth}
    \includegraphics[width=\textwidth]{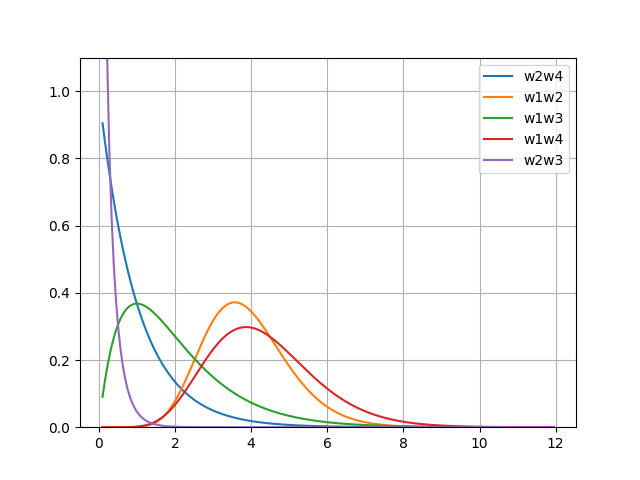}
    \caption{Posterior density at $t_5$}
\end{subfigure}
\begin{subfigure}[t]{0.42\textwidth}
    \includegraphics[width=\textwidth]{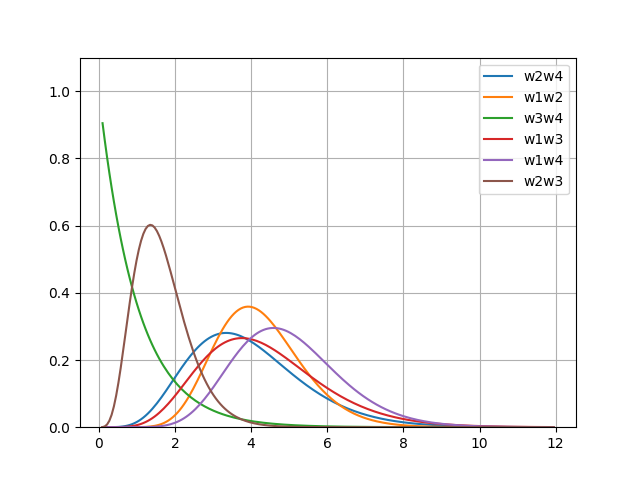}
    \caption{Posterior density at $t_6$}
\end{subfigure}
\vfill
\begin{subfigure}[t]{0.42\textwidth}
    \includegraphics[width=\textwidth]{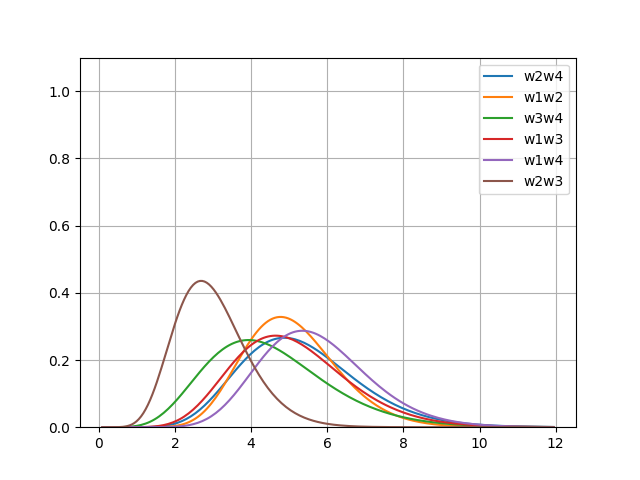}
    \caption{Posterior density at $t_7$}
\end{subfigure}
\begin{subfigure}[t]{0.42\textwidth}
    \includegraphics[width=\textwidth]{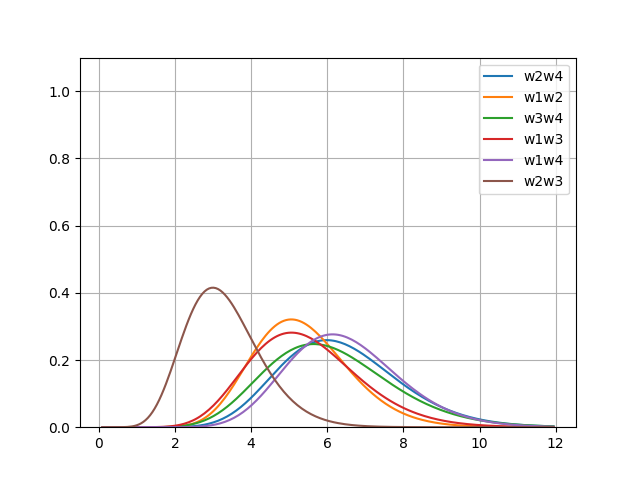}
    \caption{Posterior density at $t_8$}
\end{subfigure}
\vfill
\begin{subfigure}[t]{0.42\textwidth}
    \includegraphics[width=\textwidth]{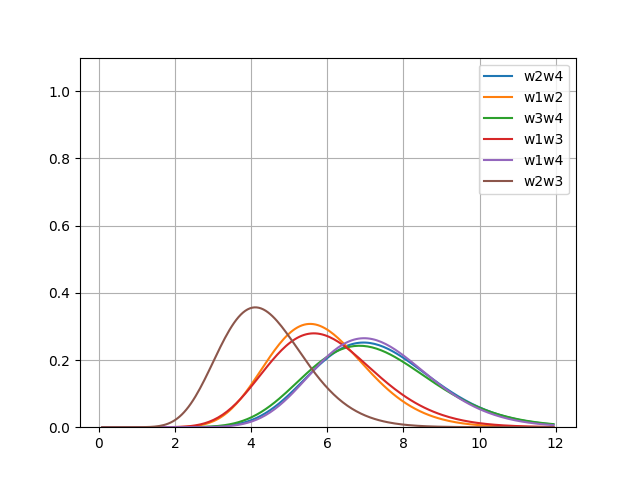}
    \caption{Posterior density at $t_9$}
\end{subfigure}
\begin{subfigure}[t]{0.42\textwidth}
    \includegraphics[width=\textwidth]{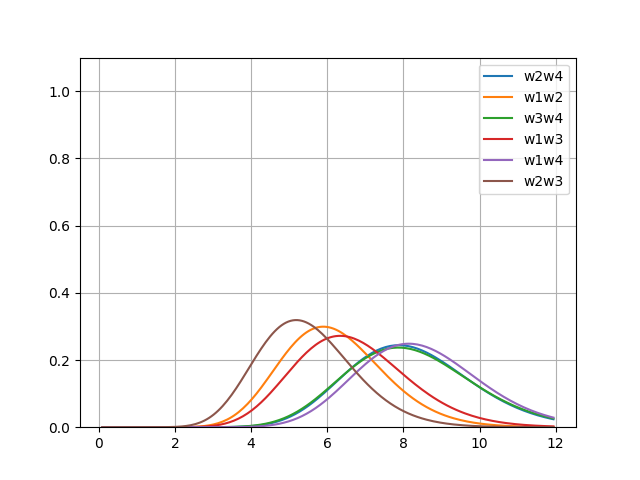}
    \caption{Posterior density at $t_{10}$}
\end{subfigure}

\caption{Evolution of $\varphi_{i,j,t_k}$ from time $t_3$ to $t_{10}$ represented through their posterior densities.}
\label{fig:posterior_densities}
\end{figure}

\printbibliography